\documentclass[preprint,3p,times]{elsarticle}
\usepackage{amsmath,amsfonts,amssymb,amsthm}
\usepackage[bottom]{footmisc}
\usepackage{subfig}
\usepackage{graphicx}
  \graphicspath{{./pics/}}
  \DeclareGraphicsExtensions{.pdf,.jpg,.png}
\usepackage[makeroom]{cancel}
\usepackage{tikz}
  \usetikzlibrary{arrows,automata,shapes}
  \usetikzlibrary{positioning,backgrounds}
  \usetikzlibrary{decorations.pathmorphing}
  \usetikzlibrary{decorations.markings}
  \usetikzlibrary{shapes.geometric}
  \usetikzlibrary{chains}
  \usetikzlibrary{fit}
  \usetikzlibrary{shapes}
  \usetikzlibrary{calc,intersections}
\usepackage{wasysym}
\usepackage{algorithm} 
\usepackage{algpseudocode}
\usepackage[textwidth=1.9cm,color=green,textsize=footnotesize]{todonotes}

\newtheorem{innercustomthm}{Theorem}
\newenvironment{customthm}[1]
  {\renewcommand\theinnercustomthm{#1}\innercustomthm}
  {\endinnercustomthm}

\newcommand{\floor}[1]{\left\lfloor #1 \right\rfloor}   
\newcommand{\eps}{\varepsilon}
\newcommand{\preq}{\preccurlyeq}
\newcommand{\uu}{u}
\newcommand{\jjj}{{\floor{\log i}}}
\newcommand{\myalpha}[2]{\alpha_{#1,#2}}
\DeclareMathOperator{\alp}{\rm alph}
\DeclareMathOperator{\0}{\bf 0}
\DeclareMathOperator{\1}{\bf 1}

\newtheorem{thm}{Theorem}
\newtheorem{lem}[thm]{Lemma}
\newtheorem{prop}[thm]{Proposition}
\newtheorem{cor}[thm]{Corollary}
\newdefinition{rmk}{Remark}
\newdefinition{example}{Example}
\newproof{pf}{Proof}
\newproof{pot}{Proof of Theorem~\ref{thm03}}

\journal{}

\begin{document}
\begin{frontmatter}

\title{Alternating Towers and Piecewise Testable Separators}

\author[cu]{\v St\v ep\'an Holub\fnref{sh}}
\ead{holub@karlin.mff.cuni.cz}

\author[tud,mu]{Tom\'{a}\v{s} Masopust\corref{cor1}\fnref{tm}}
\ead{masopust@math.cas.cz}

\author[tud]{Micha\"el~Thomazo\fnref{mt}}
\ead{michael.thomazo@tu-dresden.de}

\fntext[sh]{Research supported by the Czech Science Foundation grant number 13-01832S}
\address[cu]{Charles University, Sokolovsk\'a 83, 175 86 Praha, Czech Republic}

\fntext[tm]{Research supported by the DFG in grant KR~4381/1-1}
\address[tud]{TU Dresden, Germany}
\address[mu]{Institute of Mathematics, Czech Academy of Sciences, {\v Z}i{\v z}kova 22, 616 62 Brno, Czech Republic}

\fntext[mt]{Research supported by the Alexander von Humboldt Foundation}

\cortext[cor1]{Corresponding author}

\begin{abstract}
  Two languages are separable by a piecewise testable language if and only if there exists no infinite tower between them. An infinite tower is an infinite sequence of strings alternating between the two languages such that every string is a subsequence (scattered substring) of all the strings that follow. For regular languages represented by nondeterministic finite automata, the existence of an infinite tower is decidable in polynomial time. In this paper, we investigate the complexity of a particular method to compute a piecewise testable separator. We show that it is closely related to the height of maximal finite towers, and provide the upper and lower bounds with respect to the size of the given nondeterministic automata. Specifically, we show that the upper bound is polynomial with respect to the number of states with the cardinality of the alphabet in the exponent. Concerning the lower bound, we show that towers of exponential height with respect to the cardinality of the alphabet exist. Since these towers mostly turn out to be sequences of prefixes, we also provide a comparison with towers of prefixes.
\end{abstract}


\end{frontmatter}

\section{Introduction}
  The separation problem appears in many disciplines of mathematics and computer science, such as algebra~\cite{Almeida2008486,AlmeidaZ-ita97,PerrinPin}, logic~\cite{PlaceZ_icalp14,PlaceZ14}, formal languages~\cite{Choffrut2007,mfcsPlaceRZ13}, learning theory~\cite{learning}, and recently also in databases and query answering~\cite{icalp2013}. The latter topic is our original motivation we have investigated since the work of~\cite{icalp2013}. In that work, the motivation comes from two practical problems. The first is to make XML Schema human-readable. XML Schema is a schema language for XML that is widely accepted and supported, but is rather machine-readable than human-readable. Of course, it increases the expressiveness of Document Type Definition~\cite{xml}, but this increase goes hand in hand with loss of simplicity. Recently, the BonXai schema language has been proposed in~\cite{MartensNNS-vldb12} as an attempt to design a human-readable schema language. The BonXai schema is a set of rules $L_i \to R_i$, where $L_i$ and $R_i$ are regular expressions. An XML document (unranked tree) belongs to the language of the schema if for every node the labels of its children form a string belonging to $R_k$ and its ancestors form a string belonging to $L_k$, see~\cite{MartensNNS-vldb12} for more details. The problem we faced when translating (the finite automaton embedded in) an XML Schema Definition into an equivalent BonXai schema was that the automatically generated regular expressions $L_i$ were not human-readable. Therefore, we restricted the considered regular expressions to some ``simple'' variants, e.g., piecewise testable languages~\cite{icalp2013}.

  The second motivation comes from the observation that regular expressions are used to match paths between nodes in a graph, hence their efficient evaluation is relevant in database and knowledge-base systems. However, the graph (database, knowledge base) can be so huge that the exact evaluation is not feasible in a reasonable time~\cite{LosemannM-pods12,PerezAG-jws10}. As a solution, the expression could be rewritten to (an)other expression(s) that can be evaluated efficiently. For instance, it could be rewritten to two expressions $r_{+}$ and $r_{-}$ defining the languages that should and should not be matched in the answer. The question is then whether there exists such a ``simple'' query, and how to obtain it.

  It is not hard to see that the previous problems reduce to the language separation problem. Given two languages $K$ and $L$ and a family of languages $\mathcal{F}$, the problem asks whether there exists a language $S$ in $\mathcal{F}$ such that $S$ includes one of the languages $K$ and $L$, and is disjoint with the other. Recently, it has been independently shown in~\cite{icalp2013} and~\cite{mfcsPlaceRZ13} that the separation problem for two regular languages represented by NFAs and the family of piecewise testable languages is decidable in polynomial time with respect to both the number of states and the size of the alphabet. It should be noted that an algorithm polynomial with respect to the number of states and exponential with respect to the size of the alphabet has been known in the literature~\cite{Almeida-jpaa90,AlmeidaZ-ita97}. In~\cite{icalp2013}, the separation problem has been shown to be equivalent to the non-existence of an infinite tower between the languages. Namely, the languages have been shown separable by a piecewise testable language if and only if there does not exist an infinite tower between them. In~\cite{mfcsPlaceRZ13}, another technique has been used to prove the polynomial time bound for the decision procedure, and a doubly exponential upper bound on the index of the separating piecewise testable language. This information can be further used to construct a separating piecewise testable language with the complexity, in general, exponential with respect to the index.
  
  However, there exists a simple method (in the meaning of description) to decide the separation problem and to compute a separating piecewise testable language, whose running time corresponds to the height of the highest finite tower. This method is the original work of~\cite{pc2013} and is recalled in Section~\ref{secAlg}. The relationship between the complexity and the height of towers has motivated the study of this paper to investigate the upper and lower bounds on the height of finite towers in the presence of no infinite towers. So far, to the best of our knowledge, the only published result in this direction is a paper by Stern~\cite{Stern-tcs85}, who provided an exponential upper bound $2^{{|\Sigma|}^2 N}$ on the height of alternating towers between a piecewise testable language and its complement, where $N$ denotes the number of states of the minimal deterministic finite automaton. 

  Our contribution in this paper, which is a major revision and extension of~\cite{mfcs2014}, are the upper and lower bounds on the height of maximal finite towers between two regular languages represented by nondeterministic finite automata. Since the existence of towers of arbitrary height implies the existence of an infinite tower~\cite{icalp2013}, we restrict our attention only to the case where no infinite tower exists between the languages. We prove that the upper bound is polynomial with respect to the number of states, but exponential with respect to the size of the alphabet (Theorem~\ref{thm03}). Concerning the lower bound, we first improve the previous result showing that the bound is tight for binary regular languages up to a linear factor (Theorem~\ref{thm:quadratic}). The main result then shows that we can achieve an exponential lower bound for NFAs with respect to the size of the alphabet (Theorems~\ref{thm:exp} and~\ref{thm:2exp}). The lower bound for DFAs is discussed in Theorems~\ref{thm:expdfa} and~\ref{determinisation}. Since our towers for NFAs are in fact sequences of prefixes, we investigate the towers of prefixes in Section~\ref{TofPref}. We prove tight upper bounds on the height of towers of prefixes in Theorem~\ref{thm:dfas} and Corollary~\ref{cor:nfas}, provide a pattern that characterizes the existence of an infinite tower of prefixes (Theorem~\ref{patern}), and show that the problem is NL-complete for both NFAs and DFAs (Theorem~\ref{nl-complete} and Corollary~\ref{nl-completeCor}). Finally, Section~\ref{RelRes} provides an overview of related results. To complete it, we prove that the piecewise-testability problem is PSPACE-complete for NFAs (Theorem~\ref{pspace-complete}) and that separability of regular languages (represented by NFAs or DFAs) by piecewise testable languages is P-complete (Theorem~\ref{p-complete}).

\section{Preliminaries}
  We assume that the reader is familiar with automata and formal language theory~\cite{lawson2003finite,RozSal,sipser}. The cardinality of a set $A$ is denoted by $|A|$ and the power set of $A$ by $2^A$. An alphabet $\Sigma$ is a finite nonempty set. The elements of an alphabet are called letters. The free monoid generated by $\Sigma$ is denoted by $\Sigma^*$. A string over $\Sigma$ is any element of $\Sigma^*$. The empty string is denoted by $\eps$. For a string $w\in\Sigma^*$, $\alp(w)\subseteq\Sigma$ denotes the set of all letters occurring in $w$, and $|w|_a$ denotes the number of occurrences of letter $a$ in $w$.

  \paragraph{Automata}
  A {\em nondeterministic finite automaton\/} (NFA) is a quintuple $M = (Q,\Sigma,\delta,Q_0,F)$, where $Q$ is the finite nonempty set of states, $\Sigma$ is the input alphabet, $Q_0\subseteq Q$ is the set of initial states, $F\subseteq Q$ is the set of accepting states, and $\delta:Q\times\Sigma\to 2^Q$ is the transition function. The transition function is extended to the domain $2^Q\times\Sigma^*$ in the usual way. The language {\em accepted\/} by $M$ is the set $L(M) = \{w\in\Sigma^* \mid \delta(Q_0, w) \cap F \neq\emptyset\}$.   
  A {\em path\/} $\pi$ from a state $q_0$ to a state $q_n$ under a string $a_1a_2\cdots a_{n}$, for some $n\ge 0$, is a sequence of states and input letters $q_0, a_1, q_1, a_2, \ldots, q_{n-1}, a_{n}, q_n$ such that $q_{i+1} \in \delta(q_i,a_{i+1})$, for all $i=0,1,\ldots,n-1$. The path $\pi$ is {\em accepting\/} if $q_0\in Q_0$ and $q_n\in F$, and it is {\em simple\/} if the states $q_0,q_1,\ldots,q_n$ are pairwise distinct.
  The number of states on the longest simple path in $M$ is called the {\em depth\/} of the automaton $M$.
  We use the notation $q_0 \xrightarrow{a_1a_2\cdots a_{n}} q_{n}$ to denote that there exists a path from $q_0$ to $q_n$ under the string $a_1a_2\cdots a_{n}$.
  The NFA $M$ has a {\em cycle over an alphabet $\Gamma\subseteq\Sigma$\/} if there exists a state $q$ and a string $w$ over $\Sigma$ such that $\alp(w)=\Gamma$ and $q\xrightarrow{w} q$.
  
  The NFA $M$ is {\em deterministic\/} (DFA) if $|Q_0|=1$ and $|\delta(q,a)|=1$ for every $q$ in $Q$ and $a$ in $\Sigma$. We identify singleton sets with their elements and write $q$ instead of $\{q\}$. The transition function $\delta$ is then a map from $Q\times\Sigma$ to $Q$ that is extended to the domain $Q\times\Sigma^*$ in the usual way. Two states of a DFA are {\em distinguishable\/} if there exists a string $w$ that is accepted from one of them and rejected from the other, otherwise they are {\em equivalent}. A DFA is {\em minimal\/} if all its states are reachable and pairwise distinguishable. 

  In this paper, we assume that all automata under consideration have no useless states, that is, every state appears on an accepting path.

  \paragraph{Embedding}
  For two strings $v = a_1 a_2 \cdots a_n$ and $w \in \Sigma^* a_1 \Sigma^* a_2 \Sigma^* \cdots \Sigma^* a_n \Sigma^*$, we say that $v$ is a {\em subsequence\/} of $w$ or that $v$ can be {\em embedded\/} into $w$, denoted by $v \preccurlyeq w$. For two languages $K$ and $L$, we say that the {\em language $K$ can be embedded into the language $L$}, denoted by $K\preccurlyeq L$, if for every string $w$ in $K$, there exists a string $w'$ in $L$ such that $w\preccurlyeq w'$. We say that a {\em string $w$ can be embedded into the language $L$}, denoted by $w\preccurlyeq L$, if $\{w\}\preccurlyeq L$.

  \paragraph{Towers}
  We define {\em (alternating subsequence) towers\/} as a generalization of Stern's alternating towers, cf.~\cite{Stern-tcs85}. For two languages $K$ and $L$ and the subsequence relation $\preccurlyeq$, we say that a sequence $(w_i)_{i=1}^{r}$ of strings is an {\em (alternating subsequence) tower between $K$ and $L$\/} if $w_1 \in K \cup L$ and, for all $i = 1, \ldots, r-1$,
  \begin{itemize}
    \itemsep0pt
    \item $w_i \preccurlyeq w_{i+1}$,
    \item $w_i \in K$ implies $w_{i+1} \in L$, and
    \item $w_i \in L$ implies $w_{i+1} \in K$.
  \end{itemize}
  
  We say that $r$ is the {\em height\/} of the tower. In the same way, we define an infinite sequence of strings to be an {\em infinite (alternating subsequence) tower between $K$ and $L$}. If the languages are clear from the context, we usually omit them. Notice that the languages are not required to be disjoint, however, if there exists a $w \in K \cap L$, then there exists an infinite tower, namely $w, w, w, \ldots$. 
  
  If we talk about a {\em tower between two automata}, we mean the tower between their languages.

  \paragraph{Piecewise testable languages}  
  A regular language is {\em piecewise testable\/} if it is a finite boolean combination of languages of the form $\Sigma^* a_1 \Sigma^* a_2 \Sigma^* \cdots \Sigma^* a_k \Sigma^*$, where $k\ge 0$ and $a_i\in \Sigma$, see~\cite{Simon1972,Simon1975} for more details.

  \paragraph{Separability}  
  Let $K$ and $L$ be two languages. A language $S$ \emph{separates $K$ from $L$\/} if $S$ contains $K$ and does not intersect $L$. Languages $K$ and $L$ are \emph{separable by a family of languages $\mathcal{F}$\/} if there exists a language $S$ in $\mathcal{F}$ that separates $K$ from $L$ or $L$ from $K$.

  \paragraph{Prefixes and towers of prefixes}
  We say that a string $v\in\Sigma^*$ is a prefix of a string $w\in\Sigma^*$, denoted by $v \le w$, if $w=vu$, for some string $u\in\Sigma^*$. A sequence $(w_i)_{i=1}^{r}$ of strings is a {\em tower of prefixes between two languages $K$ and $L$\/} if $w_1 \in K \cup L$ and, for all $i = 1,2, \ldots, r-1$,
    $w_i \le w_{i+1}$,
    $w_i \in K$ implies $w_{i+1} \in L$, and
    $w_i \in L$ implies $w_{i+1} \in K$.

\section{Relevant results}\label{RelRes}
  In this section, we first briefly summarize the results concerning piecewise testable languages and separability that are relevant to the topic of this paper.

  Piecewise testable languages were studied by Simon in his PhD thesis~\cite{Simon1972}, see also~\cite{Simon1975}. He proved that piecewise testable languages are exactly those regular languages whose syntactic monoid is $\mathcal{J}$-trivial. He also provided various characterizations of piecewise testable languages in terms of monoids, automata, etc.
  These languages found applications in algebra~\cite{Almeida2008486,AlmeidaZ-ita97}, logic~\cite{PlaceZ_icalp14,PlaceZ14}, formal languages~\cite{icalp2013,mfcsPlaceRZ13} and learning theory~\cite{learning}, to mention a few.

  The fundamental question was how to efficiently recognize whether a given regular language is piecewise testable. The solution to this problem was provided by Stern in 1985 and improved by Trahtman in 2001. Stern showed that piecewise testability of a regular language represented by a deterministic finite automaton is decidable in polynomial time~\cite{Stern85a}. He provided an $O(n^5)$ algorithm, where $n$ is the number of states. Trahtman~\cite{Trahtman2001} improved Stern's result and obtained a quadratic-time algorithm to decide piecewise testability for deterministic finite automata. In 1991, Cho and Huynh~\cite{ChoH91} proved that piecewise testability is NL-complete for deterministic finite automata. To the best of our knowledge, the precise complexity of the problem for languages represented by nondeterministic finite automata has not yet been discussed in the literature. It is not hard to see that the problem is in PSPACE. We show below that it is also PSPACE-hard.
  
  Piecewise testable languages find a growing interest in separability, namely as the separating languages.
  In 1997, Almeida and Zeitoun~\cite{AlmeidaZ-ita97} developed an algorithm based on the computation of $\mathcal{J}$-closures to decide separability of regular languages represented by deterministic finite automata by piecewise testable languages. Their algorithm is polynomial with respect to the number of states, but exponential with respect to the cardinality of the alphabet. Although the algorithm is formulated for deterministic finite automata, it can be modified for nondeterministic automata.
  In 2013, Czerwi\'nski, Martens, Masopust~\cite{icalp2013} and, independently, Place, Van Rooijen, Zeitoun~\cite{mfcsPlaceRZ13} provided polynomial-time algorithms (with respect to both the size of the state space and the cardinality of the alphabet) to decide separability of regular languages represented by nondeterministic finite automata by piecewise testable languages.
  In this section, we show that separability of regular languages represented by NFAs or DFAs by piecewise testable languages is P-complete.

  It should be mentioned that separability has also been studied for other languages. For instance, separability of context-free languages by regular languages was shown undecidable a long time ago, cf.~\cite{Hunt82a}. In particular, it is shown in~\cite{Hunt82a} that even separability of simple context-free languages (so-called s-languages) by definite languages (a strict subfamily of regular languages) is undecidable.

\subsection{The piecewise-testability problem for NFAs is PSPACE-complete}
  The {\em piecewise-testability problem\/} asks whether, given a nondeterministic finite automaton $A$ over an alphabet $\Sigma$, the language $L(A)$ is piecewise testable. Although the containment to PSPACE follows basically from the result by Cho and Huynh~\cite{ChoH91}, we prefer to provide the proof here for two reasons. First, we would like to provide an unfamiliar reader with a method to recognize whether a regular language is piecewise testable. Second, Cho and Huynh assume that the input is the minimal DFA, hence it is necessary to extend their algorithm with a non-equivalence check.
  
  \begin{prop}[Cho and Huynh~\cite{ChoH91}]\label{ChoHcharacterization}
    A regular language $L$ is not piecewise testable if and only if the minimal DFA $M$ for $L$ either (1) contains a nontrivial (non-self-loop) cycle or (2) there are three distinct states $p$, $q$, $q'$ such that there are paths from $p$ to $q$ and from $p$ to $q'$ in the graph $G(M, \Sigma(q) \cap \Sigma(q'))$, where $G(M,\Gamma)$ denotes the transition diagram of the DFA $M$ restricted to edges labeled by letters from $\Gamma$, and $\Sigma(q)=\{ a\in\Sigma \mid q\xrightarrow{a} q\}$.
  \end{prop}

  \begin{lem}
    The piecewise-testability problem for NFAs is in PSPACE.
  \end{lem}
  \begin{proof}
    Let $A=(Q,\Sigma,\delta,Q_0,F)$ be an NFA. Since the automaton is nondeterministic, we cannot directly use the algorithm by Cho and Huynh~\cite{ChoH91}. However, we can consider the DFA $A'$ obtained from $A$ by the standard subset construction. The states of $A'$ are subsets of states of $A$. Now we only need to modify Cho and Huynh's algorithm to check whether the guessed states are distinguishable.
    \begin{algorithm}
      \caption{Non-piecewise testability (symbol $\rightsquigarrow$ stands for reachability)}
      \label{alg}
      \begin{algorithmic}[1]
        \State Guess states $X,Y\subseteq Q$ of $A'$;
          \Comment{Verify property (1)}
        \If {$Q_0 \rightsquigarrow X \rightsquigarrow Y \rightsquigarrow X$}
          \State go to line~\ref{noneq};
        \EndIf

        \State Guess states $P,X,Y\subseteq Q$ of $A'$; 
          \Comment{Verify property (2)}
        \State Check $Q_0\rightsquigarrow P$, $Q_0\rightsquigarrow X$, and $Q_0\rightsquigarrow Y$;
        \State $s_1 := P$; \qquad $s_2 := P$;
        \Repeat { guess} $a,b\in \Sigma(X)\cap \Sigma(Y)$;
          \State $s_1 := \delta(s_1,a)$;
          \State $s_2 := \delta(s_2,b)$;
        \Until {$s_1 = X$ and $s_2 = Y$};

        \State Guess states $X',Y'\subseteq Q$ of $A'$ such that $X'\cap F \neq\emptyset$ and $Y'\cap F=\emptyset$; 
          \Comment {Check that $X$ and $Y$ are not equivalent}
          \label{noneq}
        \State $s_1 := X$; \qquad $s_2 := Y$; 
        \Repeat { guess} $a\in\Sigma$;         
          \State $s_1 := \delta(s_1,a)$;          
          \State $s_2 := \delta(s_2,a)$;         
        \Until {$s_1 = X'$ and $s_2 = Y'$};
          
        \State\Return 'yes';
      \end{algorithmic}\label{algpspace}
    \end{algorithm}
    
    The entire algorithm is presented as Algorithm~\ref{algpspace}. In line~1 it guesses two states, $X$ and $Y$, of $A'$ that are verified to be reachable and in a cycle in lines 2-4. If so, it is verified in lines 12-17 that the states $X$ and $Y$ are not equivalent in $A'$. If there is no cycle in $A'$, property (2) of the proposition is verified in lines 5-11, and the guessed states $X$ and $Y$ are verified to be non-equivalent in lines 12-17. 
    
    The algorithm is in NPSPACE=PSPACE~\cite{Savitch1970} and returns a positive answer if and only if $A$ does not accept a piecewise testable language. Since PSPACE is closed under complement, piecewise testability is in PSPACE.
  \end{proof}

  \begin{lem}
    The piecewise-testability problem for NFAs is PSPACE-hard.
  \end{lem}
  \begin{proof}
    We prove PSPACE-hardness by reduction from the universality problem, which is PSPACE-complete~\cite{GareyJ79}. The {\em universality problem\/} asks whether, given an NFA $A$ over $\Sigma$, the language $L(A)=\Sigma^*$.

    Let $A$ be an NFA with a single initial state $q_0$ (this is not a restriction). Check whether $L(A)=\emptyset$ (in linear time). If so, return the minimal DFA $A'$ for the non-piecewise testable language $(aa)^*$. If $L(A)\neq\emptyset$, let $x$ be a new letter, and let $d$ be a new state. We ``complete'' the automaton $A$ in the sense that if no $a$-transition is defined in a state $q$, for $a$ in $\Sigma$, we add an $a$-transition from $q$ to $d$. State $d$ contains self-loops under all letters of $\Sigma$, but not under $x$. Now, we add an $x$-transition from each state, including $d$, to the initial state $q_0$. Let $A'$ denote the resulting automaton.

    If $L(A)=\Sigma^*$, we show that the language $L(A')$ is piecewise testable by showing that $L(A')=(\Sigma\cup\{x\})^*$. Indeed, $L(A)\subseteq L(A')$, so it remains to show that every string containing $x$ is accepted by $A'$. However, let $w=w_1xw_2$, where $w_1\in(\Sigma\cup\{x\})^*$ and $w_2\in\Sigma^*$. By the construction, $w_1x$ leads the automaton back to the initial state, and $w_2$ leads the automaton to an accepting state, because $w_2\in L(A)=\Sigma^*$. Thus, $w\in L(A')$. 

    Assume that $L(A)\neq\Sigma^*$. If $L(A)=\emptyset$, then $L(A')=(aa)^*$, which is not piecewise testable. If $L(A)\neq\emptyset$, consider the minimal DFA $A''$ computed from $A'$ by the standard subset construction and minimization. The DFA has at least two states, otherwise its language is either universal or empty. Every state of $A''$ is a nonempty subset of states of $A'$ (actually it is an equivalence class of such subsets, but we pick one as a representative). The empty set is not reachable because $A'$ is complete. Let $X\neq\{q_0\}$ be a state of $A''$. Then $X$ is reachable from the initial state $\{q_0\}$, and goes back to $\{q_0\}$ under $x$, which means that there is a cycle in the minimal DFA. By (1) of Proposition~\ref{ChoHcharacterization}, the language is not piecewise testable.
  \end{proof}
  
  We have proved the following result.
  \begin{thm}\label{pspace-complete}
    The piecewise-testability problem for NFAs is PSPACE-complete.
  \end{thm}

\subsection{The separation problem is P-complete}  
  The separation problem of two regular languages by a piecewise testable language is known to be equivalent to the non-existence of an infinite tower~\cite{icalp2013} and is decidable in polynomial time~\cite{icalp2013,mfcsPlaceRZ13}. In this section, we show that the problem is P-complete. We reduce the P-complete monotone circuit value problem~\cite{limits}.

  The {\em monotone circuit value problem} (MCVP) is composed of a set of boolean variables (usually called ``gates'') $g_1$, $g_2$, \ldots, $g_n$, whose values are defined recursively by equalities of the forms $g_i = \0$ (then $g_i$ is called a $\0$-gate), $g_i = \1$ ($\1$-gate), $g_i = g_j \wedge g_k$ ($\wedge$-gate), or $g_i = g_j \vee g_k$ ($\vee$-gate), where $j, k < i$. Here $\0$ and $\1$ are symbols representing the boolean values. The aim is to compute the value of $g_n$. Let $f(i)$ be the element of $\{\wedge,\vee,\0,\1\}$ such that $g_i$ is an $f(i)$-gate. For every $\wedge$-gate and $\vee$-gate, we set $\ell(i)$ and $r(i)$ to be the indices such that $g_i=g_{\ell(i)}f(i)g_{r(i)}$ is the defining equality of $g_i$. For a $\0$-gate $g_i$, we set $f(i)=\ell(i)=r(i)=\0$, and we set $f(i) = \ell(i) = r(i) = \1$ if $g_i$ is a  $\1$-gate.

  \begin{thm}\label{p-complete}
    The separation problem of two regular languages represented by NFAs by piecewise testable languages is P-complete. It is P-complete even for regular languages represented by minimal DFAs.
  \end{thm}
  \begin{proof}
    The problem was shown to be in P in~\cite{icalp2013,mfcsPlaceRZ13}. Thus, it remains to prove P-hardness.
    
    Given an instance $g_1, g_2,\ldots, g_n$ of MCVP, we construct two minimal deterministic finite automata $A$ and $B$ using a log-space reduction and prove that there exists an infinite tower between their languages if and only if the circuit evaluates gate $g_n$ to~$\1$. The theorem then follows from the fact that non-separability of two regular languages by a piecewise testable language is equivalent to the existence of an infinite tower~\cite{icalp2013}.
    
    We first construct an automaton $A'$. Let $Q_{A'}= \{s,\0,\1,1,2,\dots,n\}$ and $F_{A'}=\{\0,\1\}$. The initial state of $A'$ is $s$ and the transition function $\delta_{A'}$ is defined by $\delta_{A'}(i,a_i)=\ell(i)$ and $\delta_{A'}(i,b_i)=r(i)$. In addition, there are two special transitions $\delta_{A'}(s,x)=n$ and $\delta_{A'}(\1,y)=s$. 
    To construct the automaton $B$, let $Q_{B}= \{q,t\}\cup\{i\mid f(i)=\wedge\}$ and $F_{B}=\{q\}$, where $q$ is also the initial state of $B$.   If $f(i)=\vee$ or $f(i)=\1$, then $\delta_{B}(t,a_{i})=\delta_{B}(t,b_{i})=t$. If $f(i)=\wedge$, then $\delta_{B}(t,a_{i})=i$ and $\delta_{B}(i,b_{i})=t$.
    Finally, we define $\delta_{B}(q,x)=t$ and $\delta_{B}(t,y)=q$. (As usual, all undefined transitions go to the unique sink states of the respective automata.) An example of this construction for the circuit $g_1=\0$, $g_2=\1$, $g_3=g_1\wedge g_2$, $g_4= g_3 \vee g_3$ is shown in Figure~\ref{figPcomp}.

    The languages $L(A')$ and $L(B)$ are disjoint, the automata $A'$ and $B$ are deterministic, and $B$ is minimal. However, the automaton $A'$ need not be minimal, since the circuit may contain gates that do not contribute to the definition of the value of $g_n$. We therefore define a minimal deterministic automaton $A$ by adding transitions into $A'$, each under a fresh letter, from $s$ to each of $1,2,\dots,n-1$, from each of $1,2,\dots,n$ to state $\0$, and from $\0$ to $\1$. No new transition is defined in $B$. 

    \begin{figure}
      \centering
      \begin{tikzpicture}[baseline,->,>=stealth,shorten >=1pt,node distance=2cm,
        state/.style={circle,minimum size=1mm,very thin,draw=black,initial text=},
        every node/.style={fill=white,font=\small},
        bigloop/.style={shift={(0,0.01)},text width=1.6cm,align=center},
        bigloopd/.style={shift={(0,-0.01)},text width=1.6cm,align=center}]
        \node[state,initial]    (1) {$s$};
        \node[state]            (2) [right of=1] {$4$};
        \node[state]            (3) [right of=2] {$3$};
        \node[state]            (5) [right of=3] {$1$};
        \node[state]            (4) [below= .5cm of 5] {$2$};
        \node[state,accepting]  (6) [right of=4] {$\1$};
        \node[state,accepting]  (7) [right of=5] {$\0$};
        \path
          (1) edge node{$x$} (2)
          (2) edge node{$a_{4},b_4$} (3)
          (3) edge node{$b_3$} (4)
              edge node{$a_3$} (5)
          (6) edge[bend left=34] node{$y$} (1)
          (5) edge node{$a_1,b_1$} (7)
          (4) edge node{$a_2,b_2$} (6) ;
      \end{tikzpicture}
      \qquad
      \begin{tikzpicture}[baseline,->,>=stealth,shorten >=1pt,node distance=2cm,
        state/.style={circle,minimum size=1mm,very thin,draw=black,initial text=},
        every node/.style={fill=white,font=\small},
        bigloop/.style={shift={(0,0.01)},text width=1.6cm,align=center},
        bigloopd/.style={shift={(0,-0.01)},text width=1.6cm,align=center}]
        \node[state,initial,accepting]    (1) {$q$};
        \node[state]                      (2) [right of=1] {$t$};
        \node[state]                      (3) [right of=2] {$3$};
        \path
          (1) edge[bend left] node{$x$} (2)
          (2) edge[loop above] node[bigloop]{$a_2,b_2$\\$a_4,b_4$} (2)
          (2) edge[bend left] node{$a_3$} (3)
          (3) edge[bend left] node{$b_3$} (2)
          (2) edge[bend left] node{$y$} (1) ;
      \end{tikzpicture}
      \caption{Automata $A'$ and $B$ for the circuit $g_1=\0$, $g_2=\1$, $g_3=g_1\wedge g_2$, $g_4= g_3 \vee g_3$.}
      \label{figPcomp}
    \end{figure}
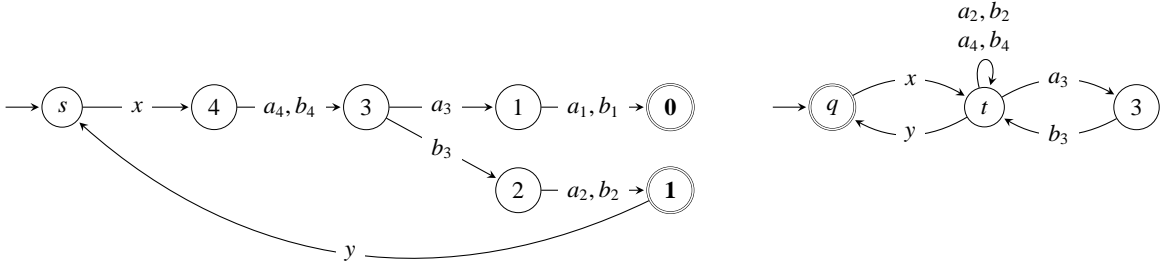

    By construction, there exists an infinite tower between the languages $L(A)$ and $L(B)$ if and only if there exists an infinite tower between $L(A')$ and $L(B)$. It is therefore sufficient to prove that the circuit evaluates gate $g_n$ to $\1$ if and only if there is an infinite tower between the languages $L(A')$ and $L(B)$. 

    (Only if) Assume that $g_n$ is evaluated to $\1$. Let $\{x,y\}\subseteq\Sigma\subseteq \{a_i,b_i\mid i=1,2,\dots,n\}\cup\{x,y\}$ be the alphabet defined as follows. The letter $a_i$ is an element of $\Sigma$ if and only if $g_i$ is evaluated to $\1$, $i$ is accessible from $n$, and either $\ell(i)=\1$ or $g_{\ell(i)}$ is evaluated to $\1$. Similarly, we have $b_i\in \Sigma$ if and only if $i$ is accessible from $n$ and $g_i$ is evaluated to $\1$, and either $r(i)=\1$ or $g_{r(i)}$ is evaluated to $\1$.  It is not hard to observe that each transition labeled by a letter $a_i$ or $b_i$ from $\Sigma$ is part of a path in $A'$ from $n$ to $\1$. Moreover, the definition of $\wedge$ implies that $a_i\in \Sigma$ if and only if $b_i\in \Sigma$ for each $i=1,2,\dots, n$ such that $f(i)=\wedge$. Notice that $B$ has a cycle from $q$ to $q$ labeled by $xa_ib_iy$ for each $i=1,2,\dots,n$ such that $f(i)=\wedge$, and also a cycle from $q$ to $q$ labeled by $xc_iy$ for each $c\in\{a,b\}$ and each $i=1,2,\dots,n$ such that $f(i)=\vee$ or $f(i)=\1$. Therefore, both automata $A'$ and $B$ have a cycle over the alphabet $\Sigma$ containing the initial and accepting states. The existence of an infinite tower follows.
    
    (If) Assume that there exists an infinite tower $(w_i)_{i=1}^{\infty}$ between $A'$ and $B$, and, for the sake of contradiction, assume that $g_n$ is evaluated to $\0$. Note that any path from $i$ to $\1$, where $g_i$ is evaluated to $\0$, must contain a state corresponding to an $\wedge$-gate that is evaluated to $\0$. In particular, this applies to any path in $A'$ accepting some $w_i$ of length at least $n+2$, since such a path contains a subpath from $n$ to $\1$. Thus, let $j$ be the smallest positive integer such that $f(j)=\wedge$, gate $g_j$ is evaluated to $\0$, and $a_j$ or $b_j$ is in $\cup_{i=1}^\infty\alp(w_i)$. The construction of $B$ implies that both $a_j$ and $b_j$ are in $\cup_{i=1}^\infty\alp(w_i)$. Since $g_j$ is evaluated to $\0$, there exists $c\in\{a,b\}$ such that the transition from $j$ under $c_j$ leads to  a state $\sigma$, where either $\sigma=\0$ or $\sigma<j$ and $g_{\sigma}$ is evaluated to $\0$. Consider a string $w_i\in L(A')$ containing $c_j$. If $w_i$ is accepted in $\1$, then the accepting path contains a subpath from $\sigma$ to $\1$, which yields a contradiction with the minimality of $j$. Therefore, $w_i$ is accepted in $\0$. However, no letter of a transition to state $\0$ appears in a string accepted by $B$, a contradiction again.
  \end{proof}

\section{Computing a piecewise testable separator}\label{secAlg}
  We now recall the simple method~\cite{pc2013} to decide the separation problem and to compute a separating piecewise testable language, and show that its running time corresponds to the height of the highest finite tower. 

  Let $L_0$ and $R_0$ be two disjoint regular languages represented by NFAs. The method first constructs a decreasing sequence of languages $\cdots\preccurlyeq R_2 \preccurlyeq L_2 \preccurlyeq R_1 \preccurlyeq L_1 \preccurlyeq R_0$ such that there exists a piecewise testable separator if and only if from some point on all the languages are empty. Then, the nonempty languages of the sequence are used to construct a piecewise testable separator.
  
  For $k\ge 1$, let $L_k=\{ w \in L_{k-1} \mid w \preq R_{k-1}\}$ be the set of all strings of $L_{k-1}$ that can be embedded into $R_{k-1}$, and let $R_k=\{ w\in R_{k-1} \mid w\preq L_k\}$, see Figure~\ref{fig2}.
  \begin{figure}
    \centering
    \begin{tikzpicture}
      \draw (1,2) node {$L_0$};
      \draw (5,2) node {$R_0$};
      \draw (1,1.5) node {$v_1\in L_1$};
      \draw (5,1) node {$R_1$};
      \draw (1,0.5) node {$v_2\in L_2$};
      \draw (5,0) node {$R_2$};
      \draw (3,0) node {$\vdots$};
      \draw[->] (1.6,1.6) -- (4.8,2);
      \draw[->] (4.8,1.1) -- (1.6,1.5);
      \draw[->] (1.6,.6) -- (4.8,1);
      \draw[->] (4.8,.05) -- (1.6,.5);
    \end{tikzpicture}
    \caption{The sequence of languages; an arrow stands for the embedding relation $\preccurlyeq$.}
    \label{fig2}
  \end{figure}
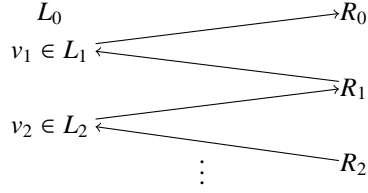
  Let $K$ be a language accepted by an NFA $A=(Q,\Sigma,\delta,Q_0,F)$. Let $down(K)$ denote the language of all subsequences of the language $K$. It is accepted by the NFA $A_{down}=(Q,\Sigma,\delta',Q_0,F)$, where $\delta'(q,a)=\delta(q,a)$ and $\delta'(q,\eps)=\bigcup_{a\in\Sigma} \delta(q,a)$. Then $L_k = L_{k-1}\cap down(R_{k-1}) = L_{0}\cap down(R_{k-1})$, and analogously for $R_k$, thus the constructed languages are regular.
  
  We now show that there exists a constant $B\ge 1$ such that $L_{B}=L_{B+1}=\ldots$, which also implies $R_B=R_{B+1}=\ldots$. Assume that no such constant exists. Then there are infinitely many strings $v_{\ell}\in L_{\ell}\setminus L_{\ell+1}$, for all $\ell\ge 1$, as depicted in Figure~\ref{fig2}. By Higman's lemma~\cite{higman}, there exist $i < j$ such that $v_i\preccurlyeq v_j$, hence $v_i \preq R_{j-1}$, which is a contradiction because $v_i \not\preq R_{i}$ and $R_{j-1}\subseteq R_i$. 

  By construction, languages $L_B$ and $R_B$ are mutually embeddable into each other, $L_B\preq R_B \preq L_B$, which describes a way how to construct an infinite tower. Thus, if there is no infinite tower, languages $L_B$ and $R_B$ must be empty. 

  The constant $B$ depends on the height of the highest finite tower. Let $(w_i)_{i=1}^{r}$ be a maximal finite tower between $L_0$ and $R_0$ and assume that $w_r$ belongs to $L_0$. In the first step, the method eliminates all strings that cannot be embedded into $R_0$, hence $w_r$ does not belong to $L_1$, but $(w_i)_{i=1}^{r-1}$ is a tower between $L_1$ and $R_0$. Thus, in each step of the algorithm, all maximal strings of all finite towers (belonging to the language under consideration) are eliminated, while the rests of towers still form towers between the resulting languages. Therefore, as long as there is a maximal finite tower, the algorithm can make another step.
  
  Assume that there is no infinite tower (that is, $L_B=R_B=\emptyset$). We use the languages computed above to construct a piecewise testable separator. For a string $w=a_1a_2\cdots a_\ell$, we define $L_w = \Sigma^* a_1 \Sigma^*a_2\Sigma^* \cdots \Sigma^* a_\ell \Sigma^*$, which is piecewise testable by definition. Let $up(L) = \bigcup_{w\in L} L_w$ denote the language of all supersequences. The language $up(L)$ is regular and its NFA is constructed from an NFA for $L$ by adding self-loops under all letters to all states, cf.~\cite{pkps2014} for more details. By Higman's Lemma~\cite{higman}, $up(L)$ can be written as a finite union of languages of the form $L_w$, for some $w\in L$, hence it is piecewise testable. For $k=1,2,\ldots,B$, we define the piecewise testable languages
  \[
    S_k = up(R_0\setminus R_k) \setminus up(L_0\setminus L_k)
 \]
  and show that $S = \bigcup_{k=1}^{B} S_k$ is a piecewise testable separator of $L_0$ and $R_0$. 
  
  More specifically, we show that $L_0\cap S_k=\emptyset$ and that $R_0\subseteq S$. To prove the former, let $w\in L_0$. If $w\in L_0 \setminus L_k$, then $w\in up(L_0 \setminus L_k)$, hence $w\notin S_k$. If $w\in L_k$ and $w\in up(R_0\setminus R_k)$, then there exists $v\in R_0\setminus R_k$ such that $v\preccurlyeq w$. However, $R_k=\{ u\in R_0 \mid u\preq L_k\}$, hence $v\in R_k$, which is a contradiction proving that $L_0\cap S_k=\emptyset$. To prove the later, we show that $R_{k-1}\setminus R_k\subseteq S_k$. Then also $R_0= \bigcup_{k=1}^{B} (R_{k-1}\setminus R_k) \subseteq S$. To show this, notice that $R_{k-1}\setminus R_k\subseteq R_0\setminus R_k \subseteq up(R_0\setminus R_k)$. If $w\in R_{k-1}$ and $w\in up(L_0\setminus L_k)$, then there exists $v\in L_0\setminus L_k$ such that $v\preccurlyeq w$. However, $L_k=\{u \in L_0 \mid u\preq R_{k-1}\}$, hence $v\in L_k$, a contradiction. Thus, we have shown that $L_0\cap S=\emptyset$ and $R_0\subseteq S$. Moreover, $S$ is piecewise testable because it is a finite boolean combination of piecewise testable languages. 

  Recall that we have shown above that the complexity of this method is closely related to the height of the highest finite tower. Therefore, in the rest of this paper, we focus on this problem. However, we would like to point out several directions of the future investigation here. 
  
  Consider the sequence of languages $R_0,L_1,R_1,\ldots,L_B,R_B$ constructed by the method. The size of the representation of the separator $S$ depends not only on the number of languages in the sequence, but also on their representation. Let $Z_k$ denote the $k$th language in the sequence. If the languages are represented by NFAs, the method gives the bound $O(n^k)$ on the size of the NFA for the language $Z_k$, where $n$ is the number of states of the NFAs for $L_0$ and $R_0$. Another option would be to use alternating finite automata (AFAs) for obvious reasons in their complexities for boolean operations. However, the problem now comes from the operation $down(Z_k)$. Even if we were able to obtain a bound $\Omega(c^k\cdot n)$ on the size of the AFA representation of the language $Z_k$, for some constant $c$, the only interesting case would be the case of $c=1$; otherwise, the problem would still be exponential. It turns out that it is not the case. It has recently been shown in~\cite{pkps2014} that the downward closure operation $down()$ for AFAs can actually require an exponential blow-up of states.
  
  Finally, even though our study in the rest of the paper shows that these representations can be doubly exponential with respect to the size of the input NFAs, practical experiments show that at some point the size of the automata does not grow anymore, but actually decreases. Indeed, the languages $L_B=R_B$ are empty, hence a representation of size $\Omega(n^{2B})$ is not needed for them.

\section{The upper bound on the height of towers}
  It was shown in~\cite{icalp2013} that there exists either an infinite tower or a constant bound on the height of any tower. We now establish an upper bound on that constant in terms of the number of states, $n$, and the cardinality of the alphabet, $m$. The bound given by Stern for minimal DFAs is $2^{m^2 n}$. Our new bound is $O(n^m)=O(2^{m\log n})$ and holds for NFAs. Thus, our bound is polynomial with respect to the number of states if the alphabet is fixed, but exponential if the alphabet grows. For $m < n\,/\log n$, the bound is less than $2^n$, but we show later that if $m = O(n)$, the height of towers may grow exponentially with respect to the number of states.

  Before we state the theorem, we recall that the depth of an automaton is the number of states on the longest simple path, thus it is bounded by the number of states of the automaton.
  
  \begin{thm}\label{thm03}
    Let $A_0$ and $A_1$ be NFAs over an alphabet $\Sigma$ of cardinality $m$ with depth at most $n$. Assume that there is no infinite tower between the languages $L(A_0)$ and $L(A_1)$. Let $(w_i)_{i=1}^r$ be a tower between $L(A_0)$ and $L(A_1)$ such that $w_i\in L(A_{i\bmod 2})$. Then $r \le \frac{n^{m+1}-1}{n-1}$.
  \end{thm}
  \begin{proof}
    First, we define some new concepts. We say that $w=v_1v_2\cdots v_k$ is a \emph{cyclic factorization} of $w$ with respect to a pair of states $(q,q')$ in an automaton $A$, if there is a sequence of states $q_0,\dots,q_{k-1},q_{k}$ such that $q_0=q$, $q_{k}=q'$, and $  q_{i-1} \stackrel {v_i} \longrightarrow q_{i}\,$, for each $i=1,2,\dots k$, and either $v_i$ is a letter, or the path $q_{i-1} \stackrel {v_i} \longrightarrow q_{i}$ contains a cycle over $\alp(v_i)$. We call $v_i$ a \emph{letter factor} if it is a letter and $q_{i-1}\neq q_i$, and a \emph{cycle factor} otherwise. The factorization is \emph{trivial} if $k=1$. Note that this factorization is closely related to the one given in~\cite{Almeida-jpaa90}, see also~\cite[Theorem~8.1.11]{AlmeidaBook}.

    We first show that if $q'\in \delta(q, w)$ in some automaton $A$ with depth $n$, then $w$ has a cyclic factorization $v_1v_2\cdots v_k$ with respect to $(q,q')$ that contains at most $n$ cycle factors and at most $n-1$ letter factors. Moreover, if $w$ does not admit the trivial factorization with respect to $(q,q')$, then $\alp(v_i)$ is a strict subset of $\alp(w)$ for each cycle factor $v_i$, $i=1,2,\dots,k$.

    Consider a path $\pi$ of the automaton $A$
    from $q$ to $q'$ labeled by a string $w$.
    Let $q_0=q$ and define the factorization  $w=v_1 v_2 \cdots v_k$ inductively
    by the following greedy strategy. 
    Assume that we have defined the factors $v_1,  v_2\ldots, v_{i-1}$
    such that $w = v_1 \cdots v_{i-1} w'$
    and $q_0 \xrightarrow{ v_1 v_2\cdots v_{i-1}} q_{i-1}$.
    The factor $v_i$ is defined as the label
    of the longest possible initial segment $\pi_i$ 
    of the path $q_{i-1}\xrightarrow{w'} q'$ 
    such that either $\pi_i$ contains a cycle over $\alp(v_i)$
    or
    $\pi_i=q_{i-1},a,q_{i}$, where $v_i=a$, thus $v_i$ is a letter.
    Such a factorization is well defined,
    and it is a cyclic factorization of $w$.

    Let $p_i$, for $i=1,2,\dots,k$, be a state
    such that the path $q_{i-1} \stackrel {v_i} \longrightarrow q_{i}$
    contains a cycle $p_i\rightarrow p_i$ over the alphabet $\alp(v_i)$ if $v_i$ is a cycle factor,
    and $p_i=q_{i-1}$ if $v_i$ is a letter factor.
    If $p_i=p_j$ with $i<j$ such that $v_i$ and $v_j$ are cycle factors,
    then we have a contradiction with the maximality of $v_i$ since
    $q_{i-1} \xrightarrow{v_i v_{i+1}\cdots v_j}  q_{j}$
    contains a cycle $p_i\rightarrow p_i$ from $p_i$ to $p_i$
    over the alphabet $\alp(v_i v_{i+1}\cdots v_j)$.
    Therefore, the factorization contains at most $n$ cycle factors.

    Note that $v_i$ is a letter factor only if the state $p_i$,
    which is equal to $q_{i-1}$ in such a case, 
    has no reappearance in the path $q_{i-1}\xrightarrow{v_i \cdots v_k} q'$.
    This implies that there are at most $n-1$ letter factors. 
    Finally, if $\alp(v_i)=\alp(w)$, then $v_i=v_1=w$
    follows from the maximality of $v_1$. 

    We now define inductively cyclic factorizations of $w_i$, such that the factorization of $w_{i-1}$ is a refinement of the factorization of $w_{i}$. Let $w_r=v_{r,1}v_{r,2}\cdots v_{r,k_r}$ be a cyclic factorization of $w_r$ defined, as described above, by some accepting path in the automaton $A_{r\bmod 2}$.
    Factorizations $w_{i-1}=v_{i-1,1}v_{i-1,2}\cdots v_{i-1,k_{i-1}}$ are defined as follows. Let 
    $
      w_{i-1}=v'_{i,1}v'_{i,2}\cdots v'_{i,k_{i}},
    $
    where $v'_{i,j}\preq v_{i,j}$, for each $j=1,2,\dots,k_{i}$. Note that such a factorization exists, since we have that $w_{i-1}\preq w_{i}$. Then $v_{i-1,1}v_{i-1,2}\cdots v_{i-1,k_{i-1}}$ is defined as a concatenation of cyclic factorizations of $v'_{i,j}$, for $j=1,2,\dots,k_{i}$, corresponding to an accepting path of $w_{i-1}$ in  $A_{i-1\bmod 2}$. The cyclic factorization of the empty string is defined as empty. Note also that a letter factor of $w_{i}$ either disappears in $w_{i-1}$, or it is  ``factored'' into a letter factor.

    In order to measure the height of a tower, we introduce
    a weight function $f$ of factors in a factorization $v_1v_2\cdots v_k$.
    First, let
    $
      g(x)= n\frac{{n^x}-1}{n-1}.
    $
    Note that $g$ satisfies $g(x+1)=n\cdot g(x)+(n-1)+1$.
    Now, let 
    $f(v_i)=1$ if $v_i$ is a letter factor, and let $f(v_i)=g(|\alp(v_i)|)$ if $v_{i}$ is a cycle factor.
    Note that, by definition, $f(\eps)=0$.
    The weight of the factorization $v_1v_2\cdots v_k$ is then defined by
    $
      W(v_1v_2\cdots v_k)=\sum_{i=1}^k f(v_i).
    $
    Let 
    $
      W_i=W(v_{i,1}v_{i,2}\cdots v_{i,k_i}).
    $
    We claim that $W_{i-1}<W_{i}$ for each $i=2,3,\dots,r$. Let $v_1v_2\cdots v_k$ be the fragment of the cyclic factorization of $w_{i-1}$ that emerged as the cyclic factorization of $v_{i,j}'\preq v_{i,j}$. If the factorization is not trivial, then, by the above analysis, 
    $
      W(v_1v_2\cdots v_k) \leq n-1 + n\cdot g(|\alp(v_{i,j})|-1)
       < g(|\alp(v_{i,j})|) = f(v_{i,j}). 
    $
    Similarly, we have $f(v_{i,j}')<f(v_{i,j})$ if $|\alp(v_{i,j}')|<|\alp(v_{i,j})|$. Altogether, we have $W_{i-1}<W_{i}$ as claimed, unless 
			$k_{i-1}=k_{i}$,
			the factor $v_{i-1,j}$ is a letter factor if and only if $v_{i,j}$ is a letter factor, and
			$\alp(v_{i-1,j})= \alp(v_{i,j})$ for all $j=1,2,\dots,k_i$.
		Assume that such a situation takes place. We show that it leads to an infinite tower. Let $L$ be the language of strings $z_1z_2\cdots z_{k_i}$ such that $z_j=v_{i,j}$ if $v_{i,j}$ is a letter factor, and $z_j\in (\alp(v_{i,j}))^*$ if $v_{i,j}$ is a cycle factor. Since $w_i\in L(A_{i\bmod 2})$ and $w_{i-1}\in L(A_{i-1\bmod 2})$, the definition of a cycle factor implies that, for each $z\in L$, there exist $z'\in L(A_0)\cap L$ such that $z\preq z'$ and $z''\in L(A_1)\cap L$ such that $z\preq z''$. The existence of an infinite tower follows. We have therefore proved that $W_{i-1}<W_{i}$. 

    The proof is completed, since $W_r\leq f(w_r)\leq g(m)$, $W_1\geq 0$, and the bound in the claim is equal to $g(m)+1$.
  \end{proof}

\section{The lower bound on the height of towers}
  It was shown in~\cite{mfcs2014} that the upper bound is tight for binary regular languages up to a linear factor. Namely, it was shown that for every odd positive integer $n$, there exist two binary NFAs with $n-1$ and $n$ states having a tower of height $n^2-4n+5$ and no infinite tower. The following theorem improves this result.
  \begin{thm}\label{thm02}
  \label{thm:quadratic}
    For every even positive integer $n$, there exists a binary NFA with $n$ states and a binary DFA with $n$ states having a tower of height $n^2-n+1$ and no infinite tower.
  \end{thm}
  \begin{proof}
    Let $n$ be an even number and define the automata $A_0$ and $A_1$ with $n$ states as depicted in Figure~\ref{ex01_1}.
    The NFA $A_0=(\{1,2,\ldots,n\},\{a,b\},\delta_0,\{1,2,\ldots,n-1\},\{n-1\})$ consists 
    of an $a$-path through states $1,2,\ldots,n-1$, 
    of self-loops under $b$ in all states $1, 2, \ldots, n-2$,
    and of a $b$-cycle from $n-1$ to $n$ and back to $n-1$.
    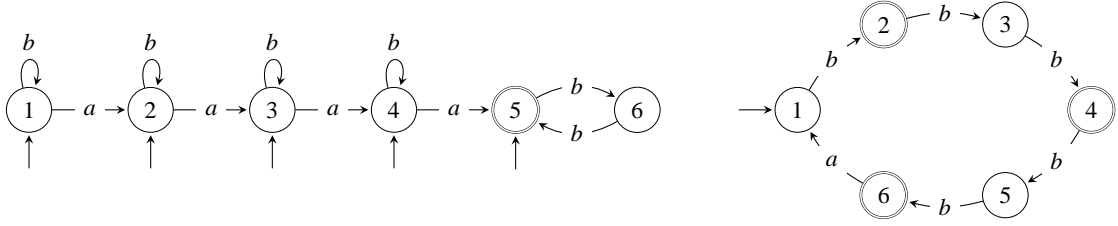
\begin{figure}
      \centering
      \begin{tikzpicture}[baseline,->,>=stealth,shorten >=1pt,node distance=1.6cm,
        state/.style={circle,minimum size=1mm,very thin,draw=black,initial text=},
        every node/.style={fill=white,font=\small},
        bigloop/.style={shift={(0,0.01)},text width=1.6cm,align=center},
        bigloopd/.style={shift={(0,-0.01)},text width=1.6cm,align=center}]
        \node[state,initial below]            (1) {$1$};
        \node[state,initial below]            (2) [right of=1] {$2$};
        \node[state,initial below]            (3) [right of=2] {$3$};
        \node[state,initial below]            (4) [right of=3] {$4$};
        \node[state,accepting,initial below]  (5) [right of=4] {$5$};
        \node[state]                          (6) [right of=5] {$6$};
        \path
          (1) edge node {$a$} (2)
          (2) edge node {$a$} (3)
          (3) edge node {$a$} (4)
          (4) edge node {$a$} (5)
          (1) edge[loop above] node[bigloop] {$b$} (1)
          (2) edge[loop above] node[bigloop] {$b$} (2)
          (3) edge[loop above] node[bigloop] {$b$} (3)
          (4) edge[loop above] node[bigloop] {$b$} (4)
          (5) edge[bend left] node {$b$} (6)
          (6) edge[bend left] node {$b$} (5);
      \end{tikzpicture}
      \quad\quad
      \begin{tikzpicture}[baseline,->,>=stealth,shorten >=1pt,node distance=1.6cm,
        state/.style={circle,minimum size=1mm,very thin,draw=black,initial text=},
        every node/.style={fill=white,font=\small},
        bigloop/.style={shift={(0,0.01)},text width=1.6cm,align=center}]
        \node[state,initial]    (1) {$1$};
        \node[state,accepting]  (2) [above right of=1]{$2$};
        \node[state]            (3) [right of=2] {$3$};
        \node[state, accepting] (4) [below right of=3] {$4$};
        \node[state]            (5) [below left of=4] {$5$};
        \node[state, accepting] (6) [left of=5] {$6$};
        
        \path
          (1) edge[bend left=15] node {$b$} (2)
          (2) edge[bend left=15] node {$b$} (3)
          (3) edge[bend left=15] node {$b$} (4)
          (4) edge[bend left=15] node {$b$} (5)
          (5) edge[bend left=15] node {$b$} (6)
          (6) edge[bend left=15] node {$a$} (1);
       \end{tikzpicture}
      \caption{Automata $A_0$ (left) and $A_1$ (right); $n=6$.}
      \label{ex01_1}
      \label{ex01_2}
    \end{figure}

    The DFA $A_1=(\{1,2,\ldots,n\},\{a,b\},\delta_1,1,\{2,4,\dots,n\})$ consists
    of a $b$-path through states $1,2,\ldots,n$ 
    and of an $a$-transition from state $n$ to state $1$. All even states are accepting.
    
    Consider the string $w=(b^{n-1}a)^{n-2}b^{n}$. Note that $A_0$ accepts all prefixes of $w$ ending with an even number of $b$s, including those ending with $a$. In particular, $w\in L(A_0)$. On the other hand, the automaton $A_1$ accepts all prefixes of $w$ ending with an odd number of $b$s. Hence the maximum height of a tower between $L(A_0)$ and $L(A_1)$ is at least $n(n-2) + n + 1 = n^2 - n + 1$.
    
    To see that there is no infinite tower between the languages, notice that any string in $L(A_0)$ contains at most $n-2$ occurrences of letter $a$. As the languages are disjoint (they require a different parity of the $b$-tail), any infinite tower has to contain a string from $L(A_1)$ of length more than $n \cdot (n-1)$. But any such string of $L(A_1)$ contains at least $n-1$ occurrences of letter $a$, hence it cannot be embedded into any string of $L(A_0)$. This means that there cannot be an infinite tower.
  \end{proof}

  Assuming that the alphabet is fixed, it was shown in~\cite{mfcs2014} that for every $n\ge 1$, there exist two NFAs with at most $n$ states over a four-letter alphabet having a tower of height $\Omega(n^3)$ and no infinite tower. However, since our original motivation comes from the XML databases, we are more interested in the case where the alphabet grows with the size of the automata.

  \begin{thm}\label{thm05}
  \label{thm:exp}
    For every $m\ge 0$, there exist an NFA with $m+1$ states and a DFA with two states over an alphabet of cardinality $m+1$ having a tower of height $2^{m+1}$ and no infinite tower.
  \end{thm}
  \begin{proof}
    Let $m$ be a non-negative integer. We define a pair of NFAs $A_m$ and $B_m$ over the alphabet $\Sigma_m=\{b,a_1,a_2,\ldots,a_m\}$ with a tower of height $2^{m+1}$ between $L(A_m)$ and $L(B_m)$, and with no infinite tower. 
    
    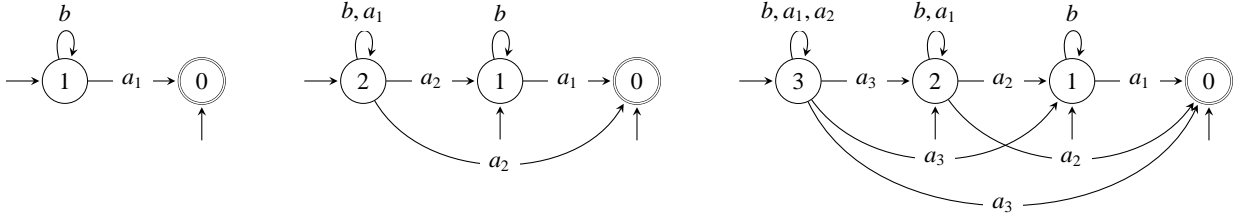
\begin{figure}
      \centering
      \begin{tikzpicture}[baseline,->,>=stealth,shorten >=1pt,node distance=1.8cm,
        state/.style={circle,minimum size=1mm,very thin,draw=black,initial text=},
        every node/.style={fill=white,font=\small},
        bigloop/.style={shift={(0,0.01)},text width=1.6cm,align=center}]
        \node[state,initial below,accepting]    (1) {$0$};
        \node[state,initial]                    (4) [left of=1] {$1$};
        \path
          (4) edge node {$a_1$} (1)
          (4) edge[loop above] node[bigloop] {$b$} (4);
      \end{tikzpicture}
      \qquad
      \begin{tikzpicture}[baseline,->,>=stealth,shorten >=1pt,node distance=1.8cm,
        state/.style={circle,minimum size=1mm,very thin,draw=black,initial text=},
        every node/.style={fill=white,font=\small},
        bigloop/.style={shift={(0,0.01)},text width=1.6cm,align=center}]
        \node[state,initial below,accepting]    (1) {$0$};
        \node[state,initial below]              (4) [left of=1] {$1$};
        \node[state,initial]                    (5) [left of=4] {$2$};
        \path
          (4) edge node {$a_1$} (1)
          (4) edge[loop above] node[bigloop] {$b$} (4)
          (5) edge[loop above] node[bigloop] {$b,a_1$} (5)
          (5) edge[bend right=60] node {$a_2$} (1)
          (5) edge node {$a_2$} (4);
      \end{tikzpicture}
      \qquad
      \begin{tikzpicture}[baseline,->,>=stealth,shorten >=1pt,node distance=1.8cm,
        state/.style={circle,minimum size=1mm,very thin,draw=black,initial text=},
        every node/.style={fill=white,font=\small},
        bigloop/.style={shift={(0,0.01)},text width=1.6cm,align=center}]
        \node[state,initial below,accepting]    (0) {$0$};
        \node[state,initial below]              (1) [left of=0] {$1$};
        \node[state,initial below]                    (2) [left of=1] {$2$};
        \node[state,initial]                    (3) [left of=2] {$3$};
        \path
          (3) edge node {$a_3$} (2)
          (2) edge node {$a_2$} (1)
          (1) edge node {$a_1$} (0)
          (3) edge[bend right=55] node {$a_3$} (1)
          (3) edge[bend right=65] node {$a_3$} (0)
          (2) edge[bend right=55] node {$a_2$} (0)
          (3) edge[loop above] node[bigloop] {$b, a_1, a_2$} (3)
          (2) edge[loop above] node[bigloop] {$b, a_1$} (2)
          (1) edge[loop above] node[bigloop] {$b$} (1);
      \end{tikzpicture}
      \caption{Automata $A_1$, $A_2$ and $A_3$.}
      \label{fig4}
      \label{fig6}
    \end{figure}
    
    The transition function $\gamma_m$ of the NFA $A_{m}=(\{0,1,2,\ldots,m\},\Sigma_{m},\gamma_{m},\{0,1,2,\ldots,m\},\{0\})$ consists of the self-loop under $b$ for all states but $0$, self-loops under $a_i$ for all states $j>i$, and $a_i$-transitions from state $i$ to all states $j<i$. Formally, $\gamma_m(i,b)=\{b\}$ if $i=1,2,\ldots,m$, $\gamma_m(i,a_j)=\{i\}$ if $m\geq i>j\geq 1$, and $\gamma_m(i,a_i)=\{0,1,\ldots,i-1\}$ if $m\geq i\geq 1$. The NFAs $A_1$, $A_2$ and $A_3$ are shown in Figure~\ref{fig6}. Note that $A_m$ is an extension of $A_{m-1}$, in particular, $L(A_{m-1})\subseteq L(A_m)$.

    We set $L(B_m)=\Sigma_m^*b$. The two-state NFA $B_m=(\{1,2\},\Sigma_m,\delta_m,\{1\},\{2\})$ consists of self-loops under all letters in state 1 and a $b$-transition from state 1 to the accepting state 2. The minimal DFA obtained from $B_m$ has two states, which fulfills the statement of the theorem.

    Define the string $\uu_m$ inductively by $u_0=\varepsilon$ and $\uu_{k}=\uu_{k-1}ba_{k}\uu_{k-1}$, for $0 < k \le m$. Note that $|\uu_kb|=2^{k+1}-1$. It is not hard to see that every prefix of $\uu_mb$ of odd length ends with $b$, therefore it is an element of $L(B_m)$. We now show, by induction on $m$, that every prefix of $\uu_mb$ of even length is accepted by $A_m$. Indeed, the empty string is accepted by $A_0=(\{0\},\{b\},\gamma_0,\{0\},\{0\})$, an automaton with no transitions. Consider a prefix $v$ of $\uu_mb$ of even length. If $|v|\leq 2^m-1$, then $v$ is a prefix of $\uu_{m-1}b$, hence, by the induction hypothesis, $v\in L(A_{m-1})\subseteq L(A_m)$. If $|v|>2^m-1$, then $v=\uu_{m-1}ba_mv'$, where $\uu_{m-1}b$ and $v'$ are over $\Sigma_{m-1}$. By the induction hypothesis, $v'$ is accepted in $A_m$ from a state $\ell\in\{0,1,\ldots,m-1\}$. The definition of $A_m$ yields that there is a path
    \[
      m\xrightarrow{\,\uu_{m-1}b\,} m\xrightarrow{\,a_m\,}\,\ell \xrightarrow{\ v'\,} 0\,.
    \]
    Thus, we have proved that the prefixes of $\uu_mb$ form a tower between $L(A_m)$ and $L(B_m)$ of height $2^{m+1}$.
    
    It remains to show that there is no infinite tower. Again, we can use the techniques of~\cite{icalp2013,mfcsPlaceRZ13}. However, to give a brief idea, notice that there is no infinite tower between $L(A_0)$ and $L(B_0)$. Consider a tower between $L(A_{m+1})$ and $L(B_{m+1})$. If every string of the tower belonging to $L(A_{m+1})$ is accepted from an initial state different from $m+1$, it is a tower between $L(A_{m})$ and $L(B_{m})$, hence it is finite by the induction hypothesis. Thus, if there is an infinite tower, there is also an infinite tower where all the strings of $L(A_{m+1})$ are accepted from state $m+1$. Every such string is of the form $\Sigma_{m}^* a_{m+1} v$, where $v$ is accepted from an initial state different from $m+1$. Cutting off the prefixes from $\Sigma_{m}^* a_{m+1}$ results in an infinite tower between $L(A_{m})$ and $L(B_{m})$, which is a contradiction.
  \end{proof}

  It is worth mentioning that the languages $L(A_m)$ are piecewise testable, hence they are the separators. The easiest way to see this is to transform $A_m$ to its minimal DFA and to use Trahtman's algorithm~\cite{Trahtman2001}.

  In the construction of the tower in the proof of Theorem~\ref{thm05}, the simplicity of the automaton $B_m$ is conspicuous. Introducing the structure present in the automaton $A_m$ also into $B_m$ can further quadratically increase the height of the tower. However, the nature of the exponential bound implies that the blow-up is only a constant in the exponent.
  
  \begin{thm}\label{thm:2exp}
    For every $m\ge 1$, there exist two NFAs with $m+1$ states over an alphabet of cardinality $2m$ having a tower of height $\Omega(2^{2m})$ and no infinite tower.
  \end{thm}
  \begin{proof}
    For every non-negative integer $m\ge 1$, we define a pair of NFAs $A_m'$ and $B_m'$ with $m+1$ states over the alphabet $\Sigma_m' = \Sigma_m \cup \{c_1,c_2,\ldots,c_{m-1}\}$, where $\Sigma_m = \{b,a_1,a_2,\ldots,a_{m}\}$ as in Theorem~\ref{thm:exp}, with a tower of height $2^{m}(2^{m}-1)+2$ between $L(A_m')$ and $L(B_m')$, and such that there is no infinite tower. Let $Q_m = \{0,1,\ldots,m\}$. 
    \begin{figure}
      \centering
      \begin{tikzpicture}[baseline,->,>=stealth,shorten >=1pt,node distance=1.8cm,
        state/.style={circle,minimum size=1mm,very thin,draw=black,initial text=},
        every node/.style={fill=white,font=\small},
        bigloop/.style={shift={(0,0.02)},text width=1.6cm,align=center}]
        \node[state,initial below,accepting]    (1) {$0$};
        \node[state,initial]                    (4) [left of=1] {$1$};
        \path
          (4) edge node {$a_1$} (1)
          (4) edge[loop above] node[bigloop] {$b$} (4);
      \end{tikzpicture}
      \qquad\qquad
      \begin{tikzpicture}[baseline,->,>=stealth,shorten >=1pt,node distance=1.8cm,
        state/.style={circle,minimum size=1mm,very thin,draw=black,initial text=},
        every node/.style={fill=white,font=\small},
        bigloop/.style={shift={(0,0.02)},text width=1.6cm,align=center}]
        \node[state,initial]    (1) {$1$};
        \node[state,accepting]  (2) [right of=1] {$0$};
        \path
          (1) edge[loop above] node[bigloop] {$b,a_1$} (1)
          (1) edge node {$b$} (2);
      \end{tikzpicture}
      \caption{The NFAs $A_1'$ (left) and $B_1'$ (right).}
      \label{fig5}
    \end{figure}
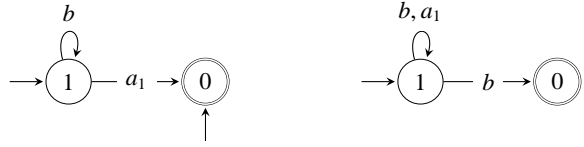
    
    The automaton $A_{m}'=(Q_{m},\Sigma_{m}',\gamma_{m}',Q_{m},\{0\})$ is obtained from the automaton $A_m$ of the proof of Theorem~\ref{thm:exp} by adding transitions from state $0$ to each state of $Q_{m}\setminus\{0\}$ under the letters $c_1,c_2,\ldots,c_{m-1}$. 
    
    The automaton $B_{m}'=(Q_{m},\Sigma_{m}',\delta_{m}',Q_{m}\setminus\{0\},\{0\})$ contains self-loops in each state $k\in\{1,2,\dots,m\}$ under the alphabet $\Sigma_m\cup\{c_1,c_2,\dots,c_{k-2}\}$, transitions under $c_{k-1}$ from state $k$ to state $i$ for each $k=2,3,\dots,m$ and $1\leq i < k$, and a $b$-transition from state $1$ to state $0$. 
    
    The automata $A_1'$ and $B_1'$ are shown in Figure~\ref{fig5} and the automata $A_3'$ and $B_3'$ in Figure~\ref{fig6b}.
    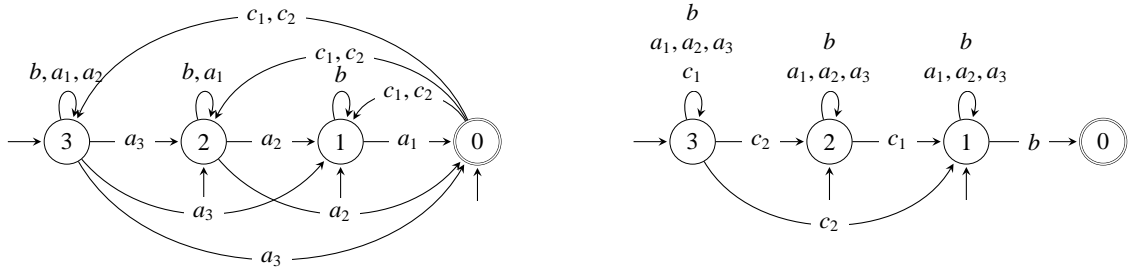
\begin{figure}[b]
      \centering
      \begin{tikzpicture}[baseline,->,>=stealth,shorten >=1pt,node distance=1.8cm,
        state/.style={circle,minimum size=1mm,very thin,draw=black,initial text=},
        every node/.style={font=\small},
        bigloop/.style={shift={(0,0.01)},text width=1.6cm,align=center}]
        \node[state,initial below,accepting]    (1) {$0$};
        \node[state,initial below]              (4) [left of=1] {$1$};
        \node[state,initial below]              (5) [left of=4] {$2$};
        \node[state,initial]                    (6) [left of=5] {$3$};
        \path
          (1) edge[bend right=60] node[fill=white] {$c_1,c_2$} (4)
          (1) edge[bend right=63] node[fill=white] {$c_1,c_2$} (5)
          (1) edge[bend right=66] node[fill=white] {$c_1,c_2$} (6)
          (4) edge node[fill=white] {$a_1$} (1)
          (4) edge[loop above] node[bigloop] {$b$} (4)
          (5) edge[loop above] node[bigloop] {$b,a_1$} (5)
          (6) edge[loop above] node[bigloop] {$b,a_1,a_2$} (6)
          (5) edge[bend right=50] node[fill=white] {$a_2$} (1)
          (5) edge node[fill=white] {$a_2$} (4)
          (6) edge[bend right=60] node[fill=white] {$a_3$} (1)
          (6) edge[bend right=50] node[fill=white] {$a_3$} (4)
          (6) edge node[fill=white] {$a_3$} (5);
      \end{tikzpicture}
      \qquad\qquad
      \begin{tikzpicture}[baseline,->,>=stealth,shorten >=1pt,node distance=1.8cm,
        state/.style={circle,minimum size=1mm,very thin,draw=black,initial text=},
        every node/.style={fill=white,font=\small},
        bigloop/.style={shift={(0,0.01)},text width=1.6cm,align=center}]
        \node[state,accepting]        (0) {$0$};
        \node[state,initial below]    (1) [left of=0]{$1$};
        \node[state,initial below]    (4) [left of=1] {$2$};
        \node[state,initial]          (5) [left of=4] {$3$};
        \path
          (1) edge node {$b$} (0)
          (4) edge node {$c_1$} (1)
          (1) edge[loop above] node[bigloop] {$b$\\$a_1,a_2,a_3$} (1)
          (4) edge[loop above] node[bigloop] {$b$\\$a_1,a_2,a_3$} (4)
          (5) edge[loop above] node[bigloop] {$b$\\$a_1,a_2,a_3$\\$c_1$} (5)
          (5) edge[bend right=60] node {$c_2$} (1)
          (5) edge node {$c_2$} (4);
      \end{tikzpicture}
      \caption{Automata $A_3'$ and $B_3'$, respectively.}
      \label{fig6b}
    \end{figure}
    
    We now define a string $\omega_m$ using the strings $\uu_i$ from the proof of Theorem~\ref{thm:exp} as follows. For $k=1,2,\ldots,m-1$, let 
    $\omega_{m,0} = \uu_{m}$, 
    $\omega_{m,k} = \omega_{m,k-1}\, c_{k}\, \omega_{m,k-1}$, and 
    $\omega_{m} = \omega_{m,m-1}\,b$.
    The string $\omega_m$ consists of $2^{m-1}$ occurrences of the string $\uu_m$ separated by some letters $c_i$. We prove by induction on $k=0,1,\dots,m-1$ that every prefix $v$ of $\omega_{m,k}$ ending with $b$ is accepted by $B_m'$ from an initial state $\ell_v \le k+1$, and every prefix of $\omega_{m,k}$ ending with some $a_i$ is accepted by $A_m'$. This is true for $\omega_{m,0}$ by Theorem 9, since $A_m'$, restricted to the alphabet $\Sigma_m$, yields $A_m$, and $B_m'$ accepts all strings over $\Sigma_m$ ending with $b$ from state $1$.
    Let $k\geq 1$ and consider the string $\omega_{m,k}$. The claim holds for prefixes of $\omega_{m,k-1}$ by induction. 
    
    Consider a prefix $v=\omega_{m,k-1}c_kv'$ ending with $a_i$. By the induction hypothesis, $A_m'$ accepts the string $\omega_{m,k-1}$ from some state $i_1$, and it accepts $v'$ from some initial state $i_2\neq 0$, since $v'$ does not begin with a $c_i$. Therefore, $A_m'$ accepts $v$ by the path $i_1\xrightarrow{\omega_{m,k-1}} 0 \xrightarrow{c_k} i_2 \xrightarrow{v'} 0$. 
    
    Let now $v=\omega_{m,k-1}c_kv'$ end with $b$. By the induction hypothesis, $v'$ is accepted by $B_m'$ from a state $\ell_{v'} \le k$, since $v'$ is a prefix of $\omega_{m,k-1}$. Moreover, $\omega_{m,k-1}$ contains no letter $c_i$, for $i>k-1$. Thus, the automaton $B_m'$ accepts $v$ by the path $k+1\xrightarrow{\,\omega_{m,k-1}\,} k+1 \xrightarrow{\,c_k\,} \ell_{v'} \xrightarrow{\ v'\,} 0$. This proves the claim.

    Since $A_m'$ accepts $\varepsilon$ and $B_m'$ accepts $\omega_m$, we have a tower of height $2^{m-1}(2^{m+1}-2)+2$, where $2^{m+1}-2$ is the length of $\uu_m$. 
    
    To prove that there is no infinite tower, we can use a similar argument as in the proof of Theorem~\ref{thm:exp}, where $B_m'$ plays the role of $A_m$ and $c_i$ the role of $a_i$, and the fact that there is no infinite tower between $A_m$ and $B_m$ over the alphabet $\Sigma_m$.
  \end{proof}
 
\subsection{The lower bound for DFAs}
  The exponential lower bounds presented above are based on NFAs. It is, however, an interesting question whether they can also be achieved for DFAs. We discuss this problem below.

  \begin{thm}\label{thm:expdfa}
    For every $n\ge 0$, there exist two DFAs with at most $n+1$ states over an alphabet of cardinality $\frac{n(n+1)}{2}+1$ having a tower of height $2^{n}$ and no infinite tower.
  \end{thm}
  \begin{proof}
    The main idea of the construction is to ``determinize'' the automata of the proof of Theorem~\ref{thm:exp}. Thus, for every non-negative integer $n$, we define a pair of deterministic automata $A_n$ and $B_n$ with $n+1$ and two states, respectively, over the alphabet $\Sigma_n=\{b\}\cup\{a_{i,j} \mid i=1,2,\dots,n;\, j=0,1,\dots,i-1\}$ with a tower of height $2^{n}$ between $L(A_n)$ and $L(B_n)$, and with no infinite tower.
    The two-state DFA $B_n=(\{1,2\},\Sigma_n,\gamma_n,1,\{2\})$ accepts all strings over $\Sigma_n$ ending with $b$ and is shown in Figure~\ref{fig4a} (right).
    \begin{figure}[t]
      \centering
      \begin{tikzpicture}[baseline,->,>=stealth,shorten >=1pt,node distance=1.8cm,
        state/.style={circle,minimum size=6mm,very thin,draw=black,initial text=},
        every node/.style={fill=white,font=\small},
        bigloop/.style={shift={(0,0.05)},text width=1.6cm,align=center}]
        \node[state,accepting]    (1) {$0$};
        \node[state]              (4) [left of=1] {$1$};
        \node[state]              (5) [left of=4] {$2$};
        \node[state,initial]              (6) [left of=5] {$3$};
        \path
          (4) edge node {$a_{1,0}$} (1)
          (4) edge[loop above] node[bigloop] {$b$\\$a_{2,0},a_{3,0}$} (4)
          (5) edge[loop above] node[bigloop] {$b,a_{1,0},$\\$a_{3,0},a_{3,1}$} (5)
          (6) edge[loop above] node[bigloop] {$b,a_{1,0},$\\$a_{2,0},a_{2,1}$} (6)
          (5) edge[bend right=55] node {$a_{2,0}$} (1)
          (5) edge node {$a_{2,1}$} (4)
          (6) edge node {$a_{3,2}$} (5)
          (6) edge[bend right=55] node {$a_{3,1}$} (4)
          (6) edge[bend right=60] node {$a_{3,0}$} (1) ;
      \end{tikzpicture}
      \qquad\qquad
      \begin{tikzpicture}[baseline,->,>=stealth,shorten >=1pt,node distance=2.3cm,
        state/.style={circle,minimum size=6mm,very thin,draw=black,initial text=},
        every node/.style={font=\small}]
        \node[state,initial]    (1) {$1$};
        \node[state,accepting]  (2) [right of=1] {$2$};
        \path
          (1) edge[loop above] node {$\Sigma_n\setminus\{b\}$} (1)
          (2) edge[loop above] node {$b$} (2)
          (1) edge[bend left] node[fill=white] {$b$} (2)
          (2) edge[bend left] node[fill=white] {$\Sigma_n\setminus\{b\}$} (1);
      \end{tikzpicture}
      \caption{The DFA $A_3$ (left) and the two-state DFA $B_n$ (right), $n\ge 0$.}
      \label{fig4a}
      \label{fig5a}
      \label{figA3var}
    \end{figure}
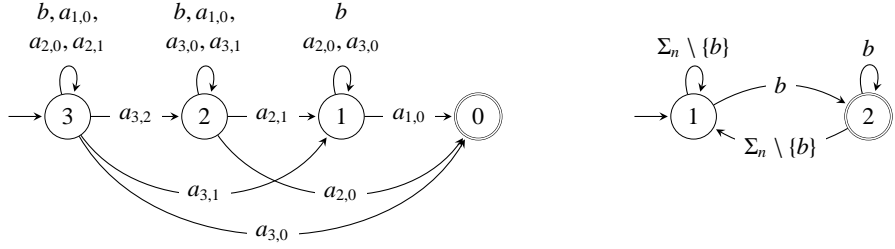
    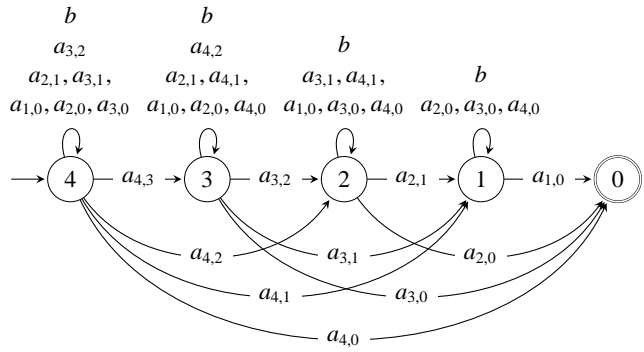
\begin{figure}[b]
      \centering
      \begin{tikzpicture}[->,>=stealth,shorten >=1pt,node distance=1.8cm,
        state/.style={circle,minimum size=6mm,very thin,draw=black,initial text=},
        every node/.style={fill=white,font=\small},
        bigloop/.style={shift={(0,0.05)},text width=1.6cm,align=center}]
        \node[state,accepting]    (1) {$0$};
        \node[state]              (4) [left of=1] {$1$};
        \node[state]              (5) [left of=4] {$2$};
        \node[state]              (6) [left of=5] {$3$};
        \node[state, initial]     (7) [left of=6] {$4$};
        \path
          (4) edge node {$a_{1,0}$} (1)
          (4) edge[loop above] node[bigloop] {$b$\\$a_{2,0},a_{3,0},a_{4,0}$} (4)
          (5) edge[loop above] node[bigloop] {$b$\\$a_{3,1},a_{4,1},$\\$a_{1,0},a_{3,0},a_{4,0}$} (5)
          (6) edge[loop above] node[bigloop] {$b$\\$a_{4,2}$\\$a_{2,1},a_{4,1},$\\$a_{1,0},a_{2,0},a_{4,0}$} (6)
          (7) edge[loop above] node[bigloop] {$b$\\$a_{3,2}$\\$a_{2,1},a_{3,1},$\\$a_{1,0},a_{2,0},a_{3,0}$} (7)
          (5) edge[bend right=55] node {$a_{2,0}$} (1)
          (5) edge node[fill=white] {\small{$a_{2,1}$}} (4)
          (6) edge node {$a_{3,2}$} (5)
          (6) edge[bend right=55] node {$a_{3,1}$} (4)
          (6) edge[bend right=60] node {$a_{3,0}$} (1)
          (7) edge node {$a_{4,3}$} (6)
          (7) edge[bend right=55] node {$a_{4,2}$} (5)
          (7) edge[bend right=60] node {$a_{4,1}$} (4)
          (7) edge[bend right=65] node {$a_{4,0}$} (1) ;
      \end{tikzpicture}
      \caption{Automaton $A_4$.}
      \label{figA4var}
    \end{figure}
    
    The ``determinization'' idea of the construction of the DFA $A_n=(\{0,1,\dots,n\},\Sigma_n,\delta_n,n,\{0\})$ is to use the automaton $A_n$ from the proof of Theorem~\ref{thm:exp}, and to eliminate nondeterminism by relabeling every transition $i\xrightarrow{a_i} j$ with a new unique letter $i\xrightarrow{a_{i,j}} j$. Then the tower of Theorem~\ref{thm:exp} is modified by relabeling the corresponding letters. However, to preserve embeddability of the new letters, several self-loops must be added. Specifically, the transition function $\delta_n$ is defined as follows.
    For every $a_{i,j}\in \Sigma_n$, define the transition $\delta_n(i,a_{i,j}) = j$.
    For every $k = 1,2,\dots,n$ and $a_{i,j}\in \Sigma_n$ such that $i\neq k$ and $j < k$, define the self-loop $\delta_n(k, a_{i,j}) = k$.
    Finally, add the self-loops $\delta_n(k,b) = k$ to every state $k=1,2,\dots,n$, see Figures~\ref{figA3var} and~\ref{figA4var} for illustration (as usual for DFAs, all undefined transitions go to a new sink state that is not depicted for simplicity). 
    
    For every $1\le k\le n$ and $0\le j < k$, let $\myalpha k j = a_{k,{j}}a_{k,j-1}\cdots a_{k,0}$, and let the strings $u_k$ be defined by $u_0=\varepsilon$ and $u_{k}=u_{k-1}b\, \myalpha{k}{k-1}\, u_{k-1}$. Note that $u_kb$ contains $2^k$ letters $b$.
    The tower of height $2^n$ between the languages $L(A_n)$ and $L(B_n)$ is the sequence $w_n(0), w_n(1), \dots, w_n(2^n-1)$, where the longest string is defined by 
    \[
      w_n=w_n(2^n-1)=\myalpha{n}{n-1}\,u_{n-1}b\in L(B_n)\,.
    \]
    Every string $w_n(2i)$ is obtained from the string $w_n(2i+1)$ by removing the last letter, which is $b$, and every string $w_n(2i-1)$ is obtained from the string $w_n(2i)$ by replacing some $\myalpha k j$ with its suffix, that is, with $\myalpha k{j'}$, where $j'\leq j$, or with the empty string, see Figure~\ref{figTower} for the case $n=3$.
    \begin{figure}
      \centering
      \begin{align*}
        w_3(0)&=\underline{a_{3,0}}\\
        w_3(1)&=\underline{a_{3,0}}\,b\\
        w_3(2)&=\underline{a_{3,1}}a_{3,0}\,b\,\underline{a_{1,0}}\\
        w_3(3)&=\underline{a_{3,1}}a_{3,0}\,b\,\underline{a_{1,0}}\,b\\
        w_3(4)&=\underline{a_{3,2}}a_{3,1}a_{3,0}\,b\,a_{1,0}\,b\,\underline{a_{2,0}}\\
        w_3(5)&=\underline{a_{3,2}}a_{3,1}a_{3,0}\,b\,a_{1,0}\,b\,\underline{a_{2,0}}\,b\\
        w_3(6)&=\underline{a_{3,2}}a_{3,1}a_{3,0}\,b\,a_{1,0}\,b\,\underline{a_{2,1}}a_{2,0}\,b\,\underline{a_{1,0}}\\
        w_3(7)&=\underline{a_{3,2}}a_{3,1}a_{3,0}\,b\,a_{1,0}\,b\,\underline{a_{2,1}}a_{2,0}\,b\,\underline{a_{1,0}}\,b
      \end{align*}
      \caption{The tower between $L(A_3)$ and $L(B_3)$. We underline transitions between different states in $A_3$.}
      \label{figTower}
    \end{figure}
  
    More explicitly, the tower is defined recursively using the towers for $A_k$ and $B_k$, where $1\leq k<n$. For any $k\geq 1$, we define
    $w_k(0)=\myalpha k 0=a_{k,0}$ and
    $w_k(1)=a_{k,0}\, b$.
    For $i\geq 2$, let  
    \begin{align}\label{eq:wki}
      w_k(i)=\myalpha k \jjj\, u_{\jjj-1}\,b\, w_{\jjj}\left(i-2^{\jjj}\right).    
    \end{align}
    By induction, we verify that this definition fits the above definition of $w_n$. Since $\floor{\log{(2^n-1)}}=n-1$, we obtain that
    $w_n(2^n-1) =\myalpha n{n-1}\, u_{n-2}\,b\, w_{n-1}(2^{n-1}-1) = \myalpha n{n-1}\, u_{n-2}\,b\, \myalpha{n-1}{n-2}\, u_{n-2}\,b = \myalpha{n}{n-1}\, u_{n-1}\,b$.
    By \eqref{eq:wki}, the relationship between $w_n(i)$ and $w_n(i+1)$ is in most cases directly induced by the relationship between $w_\jjj(i-2^\jjj)$ and $w_\jjj(i-2^\jjj+1)$. 
       A special case is when $i$ is of the form $2^\ell-1$ for some $\ell > 1$. Then $\ell-1=\floor{\log i}\neq \floor{\log (i+1)}=\ell$ and we have
    \begin{align}\label{eq:prechod}
      \begin{split}
        w_n\left(2^\ell-1\right) & = \myalpha{n}{\ell-1}\,u_{\ell-2}\,b\,w_{\ell-1}\left(2^{\ell-1}-1\right) 
                                = \myalpha{n}{\ell-1}\,u_{\ell-2}\,b\,\myalpha{\ell-1}{\ell-2}\,u_{\ell-2}\,b 
                                = \myalpha{n}{\ell-1}\,u_{\ell-1}\,b,\\
        w_n\left(2^\ell\right)   & =\myalpha n \ell\,u_{\ell-1}\,b\,w_{\ell}(0)\,,
      \end{split}
    \end{align}
    that is, $w_n(i+1)=a_{n,\ell}\,w_n(i)\, a_{\ell,0}$.

    We now prove that the sequence is the required tower. If $n=1$, the tower is $a_{1,0}$, $a_{1,0}\,b$. Let $n>1$. The definition implies that $w_n(i)$ is in $L(B_n)$ (that is, it ends with $b$) if and only if $i$ is odd. Consider $w_n(i)$ with even $i \ge 2$. 
    The path in $A_n$ defined by $w_n(i)$ can be by~(\ref{eq:wki}) decomposed as  
    \[
      n\xrightarrow{\,a_{n,\jjj}\,} \jjj\xrightarrow{\,{\myalpha{n}{\jjj-1}}\,u_{\jjj-1}\,b\,} \jjj\xrightarrow{\,w_{\jjj}\left(i-2^{\jjj}\right)\,} 0\,.
    \]
    For the second part, note that both the alphabet of ${\myalpha{n}{\jjj-1}}$ and the alphabet $\{b\}\cup \{a_{m,m'} \mid m\leq \jjj-1, m' < m\}$ of $u_{\jjj-1}\, b$ are contained in the alphabet of self-loops of state $\jjj$. The last part of the path follows by induction, since $\jjj<n$, $i-2^\jjj\leq 2^\jjj-1$, and $i-2^\jjj$ is even.

    Finally, we observe that $w_k(i)\preccurlyeq w_k(i+1)$. This follows by induction from \eqref{eq:wki} if  $\floor{\log (i+1)}=\floor{\log i}$, and from~\eqref{eq:prechod} if $i=2^\ell-1$.
  \end{proof}

  We now prove that the ``determinization'' idea of the previous theorem can be generalized. However, compared to the proof of Theorem~\ref{thm:expdfa}, the general procedure suffers from the increase of states. The reason why we do not need to increase the number of states in the proof of Theorem~\ref{thm:expdfa} is that there is an order in which the transitions/states are used/visited, and that the nondeterministic transitions are acyclic. 

  \begin{thm}\label{determinisation}
    For every two NFAs $A$ and $B$ with at most $n$ states and $m$ input letters, there exist two DFAs $A'$ and $B'$ with $O(n^2)$ states and $O(m+n)$ input letters such that there is a tower of height $r$ between $A$ and $B$ if and only if there is a tower of height $r$ between $A'$ and $B'$. In particular, there is an infinite tower between $A$ and $B$ if and only if there is an infinite tower between $A'$ and $B'$.
  \end{thm}
  \begin{proof}
    Let $A$ and $B$ be two NFAs with at most $n$ states over an alphabet $\Sigma$ of cardinality $m$. Without loss of generality, we may assume that the automata each have a single initial state. Let $Q_A$ and $Q_B$ denote their respective sets of states. We modify the automata $A$ and $B$ to obtain the DFAs $A'$ and $B'$ as follows. Let $Q_{A'}=Q_A\cup \{\sigma_{s,t} \mid s, t\in Q_A\}$ and $Q_{B'}=Q_B\cup \{\sigma_{s,t} \mid s,t\in Q_B\}$, where $\sigma_{s,t}$ are new states. We introduce a new letter $y_{t}$ for every state $t\in Q_A\cup Q_B$. It results in $O(n^2)$ states and $O(m+n)$ letters. The transition function is defined as follows. In both automata, each transition $s \xrightarrow{a} t$ is replaced with two transitions $s \xrightarrow{y_{t}} \sigma_{s,t}$ and $\sigma_{s,t} \xrightarrow{a} t$. Moreover, self-loops in all new states are added over all new letters. Note that all transitions are deterministic in $A'$ and $B'$.
    
    We now prove that if there is a tower of height $r$ between $A$ and $B$, then there is a tower of height $r$ between $A'$ and $B'$.
    Let $\left(w_i\right)_{i=1}^{r}$ be a tower between $A$ and $B$. Let
    \[
      w_i=x_{i,1}x_{i,2}\cdots x_{i,n}\,,
    \]
    where $n=|w_r|$ and $x_{i,j}$ is either a letter or the empty string such that $x_{i,j}\preccurlyeq x_{i+1,j}$, for each $i=1,2,\dots,r-1$ and $j=1,2,\dots,n$.
    For every $w_i$, we fix an accepting path $\pi_i$ in the corresponding automaton.
    Let $p_{i,j}$ be the letter $y_{t}$ where $s \xrightarrow{} t$ is the transition corresponding to $x_{i,j}$ in $\pi_i$ if $x_{i,j}$ is a letter, 
    and let $p_{i,j}$ be empty if $x_{i,j}$ is empty. 
    We define 
    \[
      w_i'=\myalpha{i}{1}\myalpha{i}{2}\cdots \myalpha{i}{n}\,,
    \]
    where $\myalpha{i}{j}=p_{i,j}p_{i-1,j}\cdots p_{1,j} a $ if $x_{i,j}=a$, and empty otherwise.
    It is straightforward to verify that  $(w_i')_{i=1}^r$ is a tower of height $r$ between $A'$ and $B'$.
    
    Let now $\left(w_i'\right)_{i=1}^r$ be a tower between $A'$ and $B'$. We show that $\left(p(w_i')\right)_{i=1}^r$ is a tower between $A$ and $B$, where $p$ is a projection erasing all new letters. Obviously, we have $p(w_i')\preccurlyeq p(w_{i+1}')$.
		We now show that if a string $w'$ is accepted by $A'$, then $p(w')$ is accepted by $A$.  
    Let $\pi'$ be the path accepting $w'$, and let $\tau_1'$, $\tau_2'$, \ldots, $\tau_k'$ denote the sequence of all transitions of $\pi'$ labeled with letters from $\Sigma$ in the order they appear in $\pi'$. 
    By construction, $\tau_i'$ is of the form $\sigma_{s,t}\xrightarrow{a} t$.
    Let $\tau_i$ be $s\xrightarrow{a} t$. We claim that $\tau_1$, $\tau_2$, \dots, $\tau_k$ is an accepting path of $p(w')$ in $A$. Indeed, the transitions $\tau_i$ are valid transitions in $A$ by the construction of $A'$. 
    Let $\tau'_{i}$, $i<k$, end in $s$. Then the successive state in $\pi'$ is $\sigma_{s,t}$, for some $t\in Q_A$, and the transition $\tau'_{i+1}$ must be $\sigma_{s,t}\xrightarrow{a} t$, for some $a\in \Sigma$, since all $\Sigma$-transitions from $\sigma_{s,t}$ end in $t$, and all other transitions are self-loops. This proves that the transitions $\tau_1$, $\tau_2$, \dots, $\tau_k$ form a path in $A$, which is accepting, since $\pi'$ is accepting and $A$ has the same initial state and accepting states as $A'$. Analogously for $B'$ and $B$. The proof is completed by~\cite[Lemma~6]{icalp2013}, which shows that there is an infinite tower if and only if there is a tower of arbitrary height. 
  \end{proof}

  A similar construction yields the following variant of the previous theorem.
  \begin{customthm}{\ref{determinisation}'}\label{determinisation2}
    For every two NFAs $A$ and $B$ with at most $n$ states and $m$ input letters, there exist two DFAs $A'$ and $B'$ with $O(mn)$ states and $O(m n)$ input letters such that there is a tower of height $r$ between $A$ and $B$ if and only if there is a tower of height $r$ between $A'$ and $B'$. In particular, there is an infinite tower between $A$ and $B$ if and only if there is an infinite tower between $A'$ and $B'$.
  \end{customthm}
  \begin{proof}
    Let $Q_{A'}=Q_A\cup \{\sigma_{a,t} \mid a\in \Sigma,\, t\in Q_A\}$ and $Q_{B'}=Q_B\cup \{\sigma_{a,t} \mid a\in \Sigma,\, t\in Q_A\}$, where $\sigma_{a,t}$ are new states. New letters are $a_t$ for every state $t\in Q_A\cup Q_B$ and every letter $a\in \Sigma$. We have $O(mn)$ states and letters. Each transition $s \xrightarrow{a} t$, in both automata, is replaced with two transitions $s \xrightarrow{a_{t}} \sigma_{a,t}$ and $\sigma_{a,t} \xrightarrow{a} t$. Self-loops in all new states are added over all new letters.

    The rest of the proof is analogous to the proof of Theorem \ref{determinisation}. We just need to slightly modify the argument that $\tau_1$, $\tau_2$, \dots, $\tau_k$ is an accepting path of $p(w')$ in $A$. Let $\sigma_{a,t}\xrightarrow{a}t$ and $\sigma_{b,q}\xrightarrow{b}q$ be two successive transitions with labels from $\Sigma$ in $\pi'$. Then the transition $\sigma_{a,t}\xrightarrow{a}t$ is in $\pi'$ necessarily followed by a transition $t\xrightarrow{b_q}\sigma_{b,q}$, which shows that $t\xrightarrow{b}q$ is a valid transition in $A$ starting in the final state of the previous transition. This was to be shown.
  \end{proof}

\section{Towers of prefixes}\label{TofPref}
  The definition of towers can be generalized from the subsequence relation to any relation on strings. However, we are particularly interested in towers of prefixes. The reason is that the lower bounds on the height of finite towers for NFAs have been obtained by towers of prefixes, that is, every string $w_i$ of the tower is a prefix of the string $w_{i+1}$. Even though we are not aware of any tower for $n$-state NFAs that would be higher then the bound on towers of prefixes for $n$-state NFAs (see below), the analysis of this section indicates that such towers exist. On the other hand, if the automata are deterministic, we show that the difference in the height between towers and towers of prefixes is exponential, cf. Theorem~\ref{thm:expdfa} and the results below.
  
  It is obvious that any tower of prefixes is also a tower (of subsequences). Hence, the non-existence of an infinite tower implies the non-existence of an infinite tower of prefixes. However, the existence of an infinite tower does not imply the existence of an infinite tower of prefixes. This can be easily seen by considering the languages $L_1=a(ba)^*$ and $L_2=b(ab)^*$. Indeed, there is no infinite tower of prefixes, since every string of $L_1$ begins with $a$ and thus cannot be a prefix of a string of $L_2$, and vice versa. But there is an infinite tower, namely, $a, bab, ababa, \ldots$. 

  In~\cite{icalp2013}, it was shown that if there exist towers of arbitrary height, then there exists an infinite tower. More precisely, this property was shown to hold for any relation that is a well quasi order (WQO), in particular for the subsequence relation. Since the prefix relation is not a WQO, the property does not hold for it in general. However, it still holds for regular languages.

  \begin{lem}\label{lem001}
    Let $K$ and $L$ be regular languages. If there are towers of prefixes of arbitrary height between $K$ and $L$, then there is an infinite tower of prefixes between them.
  \end{lem}
  \begin{proof}
    Assume that the languages are given by minimal DFAs $A=(Q_A,\Sigma,\delta_A,q_A,F_A)$ and $B=(Q_B,\Sigma,\delta_B,q_B,F_B)$ and that there is no infinite tower of prefixes between $L(A)$ and $L(B)$. In particular, the languages are disjoint. Consider the product automaton $A\times B=(Q_A\times Q_B,\Sigma,\delta,(q_A,q_B),\emptyset)$. Let $(w_i)_{i=1}^{r}$ be a tower of prefixes between $L(A)$ and $L(B)$. It defines a corresponding sequence of states $(\delta((q_A,q_B),w_i))_{i=1}^{r}$. Assume that $r > |Q_A\times Q_B|$. Then there exists a state, $(p,q)$, that appears at least twice in the sequence. Let $v_1,v_2\in\{w_1,w_2,\ldots,w_r\}$ be two different strings (of the same language $L(A)$ or $L(B)$) such that $\delta((q_A,q_B),v_j)=(p,q)$, for $j=1,2$, and $v_1\le v_2$. Because $v_1$ and $v_2$ belong to the same language, there must exist $v_1'\in\{w_1,w_2,\ldots,w_r\}$ from the other language such that $v_1\le v_1'\le v_2$. Thus, $v_1'=v_1x$ and $v_2=v_1xy$, for some nonempty strings $x$ and $y$. We can now go through the cycle and define an infinite tower of prefixes $v_1, v_1x, v_1xy, v_1xyx, \ldots$, which is a contradiction. Thus, the bound on $r$ is at most the number of states of the product automaton of the minimal DFAs.
  \end{proof}

  We now show that Lemma~\ref{lem001} does not hold for non-regular languages.
  
  \begin{example}
    Let $K=\{a,b\}^*a$ and $L=\{ a^m (ba^*)^nb \mid m > n \ge 0\}$ be two languages. Note that $K$ is regular and $L$ is non-regular context-free. The languages are disjoint, since the strings of $K$ end with $a$ and the strings of $L$ with $b$.

    For any $k\geq 1$, the strings $w_{2i+1}=a^k(ba)^i\in K$ and $w_{2(i+1)}=a^k(ba)^ib\in L$, for $i=0,1,\dots,k-1$, form a tower of prefixes between $K$ and $L$ of height $2k$.
		
    On the other hand, let $w_1, w_2, \ldots$ be a tower of prefixes between the languages $K$ and $L$. Without loss of generality, we may assume that $w_1$ belongs to $L$. Then $a^kb$ is a prefix of $w_1$, for some $k\ge 1$. It is not hard to see that $|w_i|_b<|w_{i+2}|_b$ holds for any $w_i\le w_{i+1} \le w_{i+2}$ with $w_i, w_{i+2}$ in $L$ and $w_{i+1}$ in $K$. As any string of $L$ with a prefix $a^kb$ can have at most $k$ occurrences of letter $b$, the tower cannot be infinite.
  \end{example}

\subsection{The upper and lower bounds}
  We first investigate towers of prefixes for regular languages represented by DFAs. 

  \begin{thm}\label{thm:dfas}
    Let $A$ and $B$ be two DFAs with $m$ and $n$ states that have no infinite tower of prefixes. Then the height of a tower of prefixes between $A$ and $B$ is at most $\frac{mn}{2}$, and the bound is tight.
  \end{thm}
  \begin{proof}
    Let $A=(Q_A,\Sigma,\delta_A,q_A,F_A)$ and $B=(Q_B,\Sigma,\delta_B,q_B,F_B)$, and let $X=F_A\times (Q_B\setminus F_B)$ and $Y=(Q_A\setminus F_A)\times F_B$. 
    
    We first show that $\min(|X|,|Y|)\leq \frac{mn}4$, where the equality holds only if $|X|=|Y|$. Let $0\leq \alpha,\beta \leq 1$ be such that $|F_A|=\alpha |Q_A|$ and $|F_B|=\beta |Q_B|$. Then $|X|=\alpha (1-\beta) mn$ and $|Y|=\beta (1-\alpha) mn$. The well-known inequality between the arithmetic and geometric mean yields $\frac 14 [\alpha+(1-\alpha) +\beta +(1-\beta)]=\frac 12\geq \sqrt[4]{\alpha(1-\alpha)\beta(1-\beta)}$. This implies the claim.
    
    If there is no infinite tower between $A$ and $B$, any tower of prefixes $(w_i)_{i=1}^{r}$ has to alternate between the states of $X$ and $Y$ in the product automaton $A\times B$, as in the proof of Lemma \ref{lem001}. If $r> \frac{mn}2$, then the claim above implies that at least one state is repeated, and we obtain an infinite tower of prefixes, which is a contradiction.
      
    To show that the bound is tight, we consider the automata $A_m$ and $B_m$ from the proof of Theorem~\ref{thm:exp} with a tower of prefixes of height $2^{m+1}$. The minimal DFA for $L(B_m)$ has two states. We claim that the minimal DFA for $L(A_m)$ has $2^{m+1}$ states. Indeed, consider the DFA obtained from $A_m$ by the standard subset construction. Let $Z$ and $Z'$ be two different subsets of states, and let $i\in Z\setminus Z'$. Then $a_{i}$ is accepted from $Z$, but not from $Z'$.    
    It remains to show that all subsets are reachable from the initial state $I=\{0,1,\ldots,m\}$. Recall that the $a_k$-transition in $A_m$ leads from $k$ to $j$ for all $j<k$, that there is a self-loop in $j$ under $a_k$ if $j > k$, and that the $a_k$-transition is not defined in $j$ if $j<k$. Therefore, if $k\in Z'$, the $a_k$-transition leads from $Z'$ to $Z''$, where $\{0,1,\dots,k-1\}\subseteq Z''$, $k\notin Z''$, and, for $j>k$, $j\in Z''$ if and only $j\in Z'$. Let $Z=I\setminus \{i_1,i_2,\dots,i_k\}$, where $i_1>i_2>\cdots>i_k$. It is straightforward to verify that the string $a_{i_1}a_{i_2}\cdots a_{i_k}$ leads from $I$ to $Z$.         
  \end{proof}

  Compared to towers (of subsequences), Theorem~\ref{thm:expdfa} shows that there exist towers (of subsequences) of exponential height with respect to the number of states of the input DFAs, while Theorem~\ref{thm:dfas} gives a quadratic bound on the height of towers of prefixes with respect to the number of states of the input DFAs. This shows an exponential difference between the height of towers (of subsequences) and the height of towers of prefixes.
  
  What is the situation for NFAs? An immediate consequence of the NFA-to-DFA transformation and Theorem~\ref{thm:dfas} give the following tight bound.
  \begin{cor}\label{cor:nfas}
    Given two NFAs with at most $m$ and $n$ states and with no infinite tower of prefixes. The height of a tower of prefixes is at most $2^{m+n-1}$. Moreover, the lower bound is at least $2^{m+n-2}$ for infinitely many pairs $(m,n)$.
  \end{cor}
  \begin{proof}
    The automata $A_{m-1}$ and $B_{m-1}$ in Theorem~\ref{thm:exp} have $m$ and $2$ states, respectively, and the tower is a tower of prefixes of height $2^{m}$.
  \end{proof}

  Unlike the bounds on towers (of subsequences), which depend on both the number of states and the size of the alphabet, the bounds on towers of prefixes depend only on the number of states. A natural question is whether there are any requirements on the size of the alphabet in case of automata with exponentially high towers of prefixes. The following corollary shows that the alphabet can be binary and the tower is still more than polynomial in the number of states. 	
  \begin{cor}
    There exist infinitely many pairs $(n_1,n_2)$ of integers for which there are binary NFAs with $n_1$ and $n_2$ states with no infinite tower of prefixes and with a tower of prefixes of a superpolynomial height with respect to $n_1+n_2$.
  \end{cor}
  \begin{proof}
    The property of being a tower of prefixes is preserved if the alphabet is encoded in binary. The binary code of each letter has length at most $\log m$ for an alphabet of cardinality $m$. Therefore, every transition under an original letter can be replaced by a path with at most $\log m$ new states. 
    
    Consider again the automata $A_m$ and $B_m$ from Theorem~\ref{thm:exp}. Every automaton $A_m$ has $m+1$ states and $m(m+1)$ transitions, and every automaton $B_m$ has two states and $m+2$ transitions. Encoding every letter in binary results in automata with $n_1 = O(m^2\log m)$ and $n_2 = O(m\log m)$ states, respectively, and a tower of prefixes of height $2^{m+1} = 2^{\Omega\left(\frac{\sqrt{n_1+n_2}}{\log(n_1+n_2)}\right)}$.
  \end{proof}

\subsection{The complexity analysis}
  In the rest of the paper, we study the complexity of the problem whether there exists an infinite tower of prefixes between two (nondeterministic) finite automata accepting disjoint languages.
  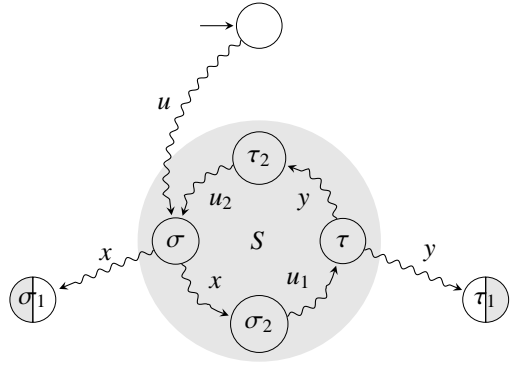
\begin{figure}[t]
    \centering
    \begin{tikzpicture}[->,>=stealth,shorten >=1pt,auto,node distance=1.1cm,
      state/.style={circle,minimum size=6mm,very thin,draw=black,initial text=}]
        \node[]       (1) {$S$};
        \node[state]  (2) [left of=1]  {$\sigma$};
        \node[state]  (3) [right of=1] {$\tau$};
        \node[state]  (4) [below of=1]  {$\sigma_2$};
        \node[state]  (5) [above of=1] {$\tau_2$};
        \path
          (2) edge[-stealth,decoration={snake,amplitude=.4mm,segment length=2mm,post length=1.9mm},decorate,bend right] node {$x$} (4)
          (4) edge[-stealth,decoration={snake,amplitude=.4mm,segment length=2mm,post length=1.9mm},decorate,bend right] node {$u_1$} (3)
          (3) edge[-stealth,decoration={snake,amplitude=.4mm,segment length=2mm,post length=1.9mm},decorate,bend right] node {$y$} (5)
          (5) edge[-stealth,decoration={snake,amplitude=.4mm,segment length=2mm,post length=1.9mm},decorate,bend right] node {$u_2$} (2) ;
        \begin{pgfonlayer}{background}
          \filldraw [line width=4mm,black!10] (1 -| 1) ellipse (1.4cm and 1.4cm);
        \end{pgfonlayer}
        \node[] (6) [left of=2]  {};
        \node[] (7) [right of=3] {};
        \node[] (8) [below left of=6]  {};
        \node[state,semicircle,draw,semicircle,rotate=90,scale=.5,fill=black!10,anchor=south] at (8) {};
        \node[state,semicircle,draw,semicircle,rotate=-90,scale=.5,anchor=south] at (8) {};
        \node[] (8) [below left of=6]  {$\sigma_1$};
        \node[] (9) [below right of=7] {};
        \node[state,semicircle,draw,semicircle,rotate=-90,scale=.5,fill=black!10,anchor=south] at (9) {};
        \node[state,semicircle,draw,semicircle,rotate=90,scale=.5,anchor=south] at (9) {};
        \node[] (9) [below right of=7] {$\tau_1$};
        \node[state,initial] (i) [above=of 5] {};
        \path
          (2) edge[-stealth,decoration={snake,amplitude=.4mm,segment length=2mm,post length=1.9mm},decorate] node[above] {$x$} (8)
          (3) edge[-stealth,decoration={snake,amplitude=.4mm,segment length=2mm,post length=1.9mm},decorate] node{$y$} (9)
          (i) edge[-stealth,decoration={snake,amplitude=.4mm,segment length=2mm,post length=1.9mm},decorate,bend right] node[above left]{$u$} (2) ;
    \end{tikzpicture}
    \caption{The pattern $(S,\sigma,\sigma_1,\sigma_2,\tau,\tau_1,\tau_2)$}
    \label{fig1}
  \end{figure}
  
  Let $A=(Q_A,\Sigma,\delta_A,q_A,F_A)$ and $B=(Q_B,\Sigma,\delta_B,q_B,F_B)$ be two NFAs. We say that $(S,\sigma,\sigma_1,\sigma_2,\tau,\tau_1,\tau_2)$ is a \emph{pattern} of the automata $A$ and $B$ if $S$ is a nontrivial (it has at least one edge) strongly connected component of the product automaton $A\times B$, and $\sigma$, $\sigma_1$, $\sigma_2$, $\tau$, $\tau_1$, $\tau_2$ are states of the product automaton such that 
   \begin{itemize}
    \itemsep0pt
    \item $\sigma_1\in F_A\times Q_B$ and $\tau_1\in Q_A\times F_B$,
    \item $\sigma, \sigma_2, \tau, \tau_2$ belong to $S$,
    \item states $\sigma_1$ and $\sigma_2$ are reachable from the state $\sigma$ under a common string,
    \item states $\tau_1$ and $\tau_2$ are reachable from the state $\tau$ under a common string, and
    \item the strongly connected component $S$ is reachable from the initial state.
   \end{itemize}
  The definition is illustrated in Figure~\ref{fig1}.
 
  The following theorem provides a characterization for the existence of an infinite tower of prefixes.
  \begin{thm}\label{patern}
    Let $A$ and $B$ be two NFAs such that $L(A)$ and $L(B)$ are disjoint. Then there is a pattern of the automata $A$ and $B$ if and only if there is an infinite tower of prefixes between $A$ and $B$.
  \end{thm}
  \begin{proof}
    First, note that the states of $F_A\times F_B$ are not reachable from the initial state $q_0=(q_A,q_B)$, because the languages are disjoint. Assume that $(S,\sigma,\sigma_1,\sigma_2,\tau,\tau_1,\tau_2)$ is a pattern of the automata $A$ and $B$. Let $u$ denote the shortest string under which state $\sigma$ is reachable from the initial state $(q_A,q_B)$. Let $x$ ($y$ resp.) be a string under which both $\sigma_1$ and $\sigma_2$ ($\tau_1$ and $\tau_2$ resp.) are reachable from $\sigma$ ($\tau$ resp.). Let $u_1$ denote the shortest string under which $\tau$ is reachable from $\sigma_2$, and $u_2$ denote the shortest string under which $\sigma$ is reachable from $\tau_2$, see Figure~\ref{fig1}. Together, we have an infinite tower of prefixes $u (x u_1 y u_2)^* (x + x u_1 y)$.

    To prove the other direction, assume that there exists an infinite tower of prefixes $(w_i)_{i=1}^{\infty}$ between the languages $L(A)$ and $L(B)$. Consider the automaton $det(A\times B)$, the determinization of $A\times B$ by the standard subset construction. A sufficiently long element of the tower defines a path $q_0\xrightarrow{u} X \xrightarrow{z_X} Y \xrightarrow{z_Y} X$ in the automaton $det(A\times B)$, such that $Y$ contains a state $(f_1,q_1)\in F_A\times (Q_B\setminus F_B)$ and $X$ contains a state $(q_2,f_2)\in (Q_A\setminus F_A)\times F_B$. 
    
    Since $Y=\delta_{A\times B}(X,z_X)$ and $X=\delta_{A\times B}(Y,z_Y)$, for every state of $X$, there exists an incoming path from an element of $Y$ labeled by $z_Y$. Similarly, for every state of $Y$, there exists an incoming path from an element of $X$ labeled by $z_X$. Thus, there are infinitely many paths from $X$ to $X$ labeled with $(z_Xz_Y)^+$ ending in state $(q_2,f_2)$. Therefore, there exists a state $(s_1,t_1)\in X$ and integers $k_1$ and $\ell_1$ such that $(s_1,t_1)\xrightarrow{(z_Xz_Y)^{k_1}} (s_1,t_1)\xrightarrow{(z_Xz_Y)^{\ell_1}} (q_2,f_2)$. Similarly, there exists a state $(s_2,t_2)\in X$ and integers $k_2$ and $\ell_2$ such that $(s_2,t_2)\xrightarrow{(z_Xz_Y)^{k_2}} (s_2,t_2)\xrightarrow{(z_Xz_Y)^{\ell_2}z_X} (f_1,q_1)$.

    As $q_0\xrightarrow{u} (s_1,t_1)$ and $q_0\xrightarrow{u} (s_2,t_2)$, state $(s_2,t_1)$ belongs to $X$. Moreover, $(s_2,t_1) \xrightarrow{(z_Xz_Y)^{k_1k_2}} (s_2,t_1)$ forms a cycle. Thus, state $(s_2,t_1)$ appears in a nontrivial strongly connected component of the automaton $A\times B$. From the above we also obtain $(s_2,t_1)\xrightarrow{(z_Xz_Y)^{\ell_1}} (s_3,f_2)$, for an $s_3$ in $Q_A$, which exists because $s_2\xrightarrow{(z_Xz_Y)^{k_2}} s_2 $ is a cycle in $A$. Similarly, we obtain $(s_2,t_1)\xrightarrow{(z_Xz_Y)^{\ell_2}z_X} (f_1,t_3)$, for a $t_3$ in $Q_B$. Notice that $(f_1,t_3)\in F_A\times (Q_B\setminus F_B)$ and $(s_3,f_2)\in (Q_A\setminus F_A)\times F_B$. 
    
    Thus, we have a pattern with $\sigma=\tau=(s_2,t_1)$, $\sigma_1=(f_1,t_3)$, $\tau_1=(s_3,f_2)$, and with $\sigma_2$ and $\tau_2$ being states of the cycle $\sigma \xrightarrow{(z_Xz_Y)^{k_1k_2}}\sigma$ satisfying $\sigma \xrightarrow{x} \sigma_2 \xrightarrow{u_1} \tau \xrightarrow{y} \tau_2 \xrightarrow{u_2} \sigma$, where $x=(z_Xz_Y)^{\ell_2}z_X$ and $y=(z_Xz_Y)^{\ell_1}$ have been obtained above, and $u_1$ and $u_2$ can be chosen as $u_1=z_Y(z_Xz_Y)^{\ell_2k_1k_2-\ell_2-1}$ and $u_2=(z_Xz_Y)^{\ell_1k_1k_2-\ell_1}$.
  \end{proof}

  We now use the previous result to study the complexity of the problem of {\em the existence of an infinite tower of prefixes}. The problem asks whether, given two (nondeterministic) finite automata accepting disjoint languages, there exists an infinite tower of prefixes between their languages.
  \begin{thm}\label{nl-complete}
    The problem of the existence of an infinite tower of prefixes is NL-complete.
  \end{thm}
  \begin{proof}
    The problem is in NL, since it is sufficient to guess the six states $\sigma$, $\sigma_1$, $\sigma_2$, $\tau$, $\tau_1$, $\tau_2$ of a pattern and then it takes several reachability tests to verify the guess, see Algorithm~\ref{alg2}. This can be done in nondeterministic logarithmic space~\cite{sipser}.
    \begin{algorithm}
      \caption{Checking the existence of a pattern of automata $A$ and $B$ (symbol $\rightsquigarrow$ stands for reachability)}
      \label{alg2}
      \begin{algorithmic}[1]
        \State Guess six states $\sigma$, $\sigma_1$, $\sigma_2$, $\tau$, $\tau_1$, $\tau_2$ such that $\sigma_1\in F_A\times Q_B$ and $\tau_1\in Q_A\times F_B$;
        \State Check $\sigma\rightsquigarrow \sigma_2 \rightsquigarrow \tau \rightsquigarrow \tau_2 \rightsquigarrow \sigma$; 
          \Comment{$\sigma, \sigma_2, \tau, \tau_2$ belong to the same SCC}
        \State Check reachability of $\sigma$ from the initial state of $A\times B$
          
        \State $k_1 := \sigma$; \qquad $k_2 := \sigma$; 
          \Comment{$\sigma_1$ and $\sigma_2$ are reachable from $\sigma$ under a common string}
        \Repeat { guess} $a\in\Sigma$;         
          \State $k_1 := \delta(k_1,a)$;          
          \State $k_2 := \delta(k_2,a)$;         
        \Until {$k_1 = \sigma_1$ and $k_2 = \sigma_2$};
      
        \State $k_1 := \tau$; \qquad $k_2 := \tau$; 
          \Comment{$\tau_1$ and $\tau_2$ are reachable from $\tau$ under a common string}
        \Repeat { guess} $a\in\Sigma$;         
          \State $k_1 := \delta(k_1,a)$;          
          \State $k_2 := \delta(k_2,a)$;         
        \Until {$k_1 = \tau_1$ and $k_2 = \tau_2$};
        \State\Return 'yes';
      \end{algorithmic}
    \end{algorithm}
    
    To prove NL-hardness we reduce the reachability problem. The {\em reachability problem\/} asks, given a directed graph $G$ and vertices $s$ and $t$, whether $t$ is reachable from $s$. The problem is NL-complete~\cite[Theorem~16.2]{papadimitriou}. 
    
    Let $G=(V,E,s,t)$ be an instance of the reachability problem. We construct the automaton $A=(V,\Sigma,\delta_A,q_0,\{q_0\})$, where $\delta_A$ is defined as the relation $E$ where every transition is given a unique label. In addition, there are two more transitions $\delta_A(q_0,a)=s$ and $\delta_A(t,b)=q_0$, for some fresh letters $a$ and $b$. Thus, $\Sigma=\{a,b\}\cup\{\ell_i \mid 1\le i\le |E|\}$. The automaton $B$ is depicted in Figure~\ref{figNL4}. It is not hard to see that there exists an infinite tower of prefixes if and only if $t$ is reachable from $s$.
    \begin{figure}
      \centering
      \begin{tikzpicture}[->,>=stealth,shorten >=1pt,auto,node distance=1.6cm,
        state/.style={circle,minimum size=6mm,very thin,draw=black,initial text=}]
          \node[state]            (2) [] {$s$};
          \node[]                 (3) [right of=2] {};
          \node[state]            (4) [below right of=3] {$t$};
          \path
            (2) edge[style={decorate, decoration={snake,post length=1.9mm}}] node{?} (4) ;
          \begin{pgfonlayer}{background}
            \filldraw [line width=4mm,rounded corners,black!10]
              (2.north -| 2.west)  rectangle (4.south -| 4.east);
          \end{pgfonlayer}
          \node[state,initial,accepting]    (1) [below left of=2] {$q_0$};
          \path
            (1) edge[bend left] node{$a$} (2)
            (4) edge[bend left=60] node{$b$} (1) ;

        \node[state,initial]    (10) [right=2cm of 4] {$i$};
        \node[state,accepting]  (20) [right of=10] {$j$};
        \path
          (10) edge[bend left] node{$a$} (20)
          (20) edge[loop above] node{\qquad$\Sigma\setminus\{a,b\}$} (20)
          (20) edge[bend left] node{$b$} (10) ;
      \end{tikzpicture}
      \caption{Automata $A$ and $B$.}
      \label{figNL4}
    \end{figure}
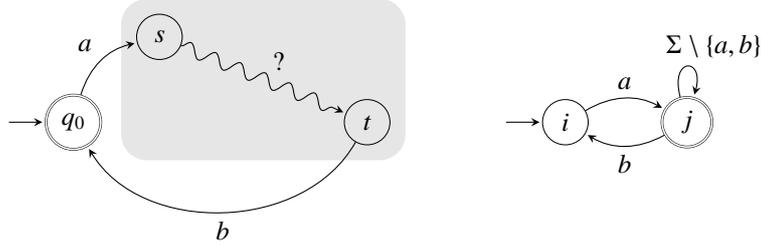
  \end{proof}
  
  \begin{cor}\label{nl-completeCor}
     The problem of the existence of an infinite tower of prefixes for minimal DFAs is NL-complete.
  \end{cor}
  \begin{proof}
    It remains to prove NL-hardness for DFAs. Consider the construction from the proof of Theorem~\ref{nl-complete}. We add a new accepting state, $f$, and new transitions under fresh letters from $q_0$ to every node of $G$, and from every node of $G$ to $f$. Transitions under these letters are undefined in $B$. Then $A$ and $B$ are minimal DFAs and there exists an infinite tower of prefixes if and only if $t$ is reachable from $s$.
  \end{proof}

\section{Conclusions}
  We have provided upper and lower bounds on the height of maximal finite towers between two regular languages represented by nondeterministic finite automata in the case there is no infinite tower. Both the upper and lower bound is exponential with respect to the size of the alphabet, which means that the algorithm of Section~\ref{secAlg} needs to handle at least exponential number of regular languages. In addition, we have shown that the exponential lower bound can be obtain not only for NFAs, but also for DFAs. And the observation that the lower bounds for NFAs are formed by sequences of prefixes has motivated the investigation of towers of prefixes.
  
  Finally, note that there is still a gap between the upper and lower bound on the height of towers. Moreover, the question how to efficiently compute a piecewise testable separator is open, including the question whether an efficient algorithm exists.

\subsubsection*{Acknowledgements}
  The authors would like to thank Stefan Borgwardt, Wojciech Czerwi\'nski, Galina Jir\'askov\'a, Markus Kr\"otzsch, and Wim Martens for valuable discussions on various parts of the paper.

\section*{\refname}
\bibliographystyle{elsarticle-harv}
\bibliography{paper}
 
\end{document}